\documentclass[journal]{IEEEtran}

\usepackage{graphicx, float}
\usepackage{amsmath}
\usepackage{amsthm}
\usepackage{amsfonts}
\usepackage{cite}
\usepackage{xcolor}
\usepackage{url}
\usepackage{booktabs}
\usepackage[subtle]{savetrees}
\newtheorem{theorem}{Theorem}[section]

\newtheorem*{remark}{Remark}

\ifCLASSOPTIONcompsoc
    \usepackage[caption=true, font=normalsize, labelfont=sf, textfont=sf]{subfig}
\else
\usepackage[caption=false, font=footnotesize]{subfig}
\fi
\IEEEoverridecommandlockouts  

\title{Optimal and resilient coordination of virtual batteries in \\ distribution feeders}
\author{Sarnaduti~Brahma$^\ast$, Nawaf~Nazir$^\ast$, Hamid~Ossareh$^\ast$, and Mads~Almassalkhi$^\ast$
\thanks{$^{\ast}$The authors are affiliated with the Department of Electrical and Biomedical Engineering, The University of Vermont, Burlington, VT 05405, USA. Support from the U.S. Department of Energy award number DE-EE0008006 is gratefully acknowledged.}}%

\date{August 2019}

\begin{document}

\maketitle

\begin{abstract}\label{abstract}
    This paper presents a novel hierarchical framework for real-time, network-admissible coordination of responsive grid resources aggregated into virtual batteries (VBs). In this context, a VB represents a local aggregation of directly controlled loads, such as smart inverters, electric water heaters, and air-conditioners.
    The coordination is achieved by solving an optimization problem to disaggregate a feeder's desired reference trajectory into constraint-aware set-points for the VBs. Specifically, a novel, provably-tight, convex relaxation of the AC optimal power flow (OPF) problem is presented to optimally dispatch the VBs to track the feeder's desired power trajectory. In addition to the optimal VB dispatch scheme, a real-time, corrective control scheme is designed, which is based on optimal proportional-integral (PI) control with anti-windup, to reject intra-feeder and inter-feeder disturbances that arise during operation of the power system.
    Simulation results conducted on a modified IEEE test system demonstrate the effectiveness of the proposed multi-layer VB coordination framework. 
\end{abstract}
\section{Introduction}\label{introduction}
\subsection{Background and Motivation}
Coordinated control of demand-side, distributed energy resources (DERs), such as grid-tied PV inverters, distributed battery storage, and thermostatically controlled loads (TCLs; e.g., water heaters and air conditioners) 
is part of the solution that supports a renewable energy future~\cite{schweppe1980homeostatic,brooks:2010demand, callaway2010achieving, brattle:2019LoadFlex2}.
Much of the recent literature on the  coordination of DERs has focused on distributed control methodologies to turn large-scale aggregations of DERs into dispatchable grid assets (similar to ~\cite{mathieu2014arbitraging, meyn2015ancillary, Almassalkhi2018}). Since the aggregation of DERs is dispatched as a single entity by a centralized coordinator and is subject to power and energy limits, the DER fleet is often referred to as a ``virtual battery" (VB)~\cite{hao2014aggregate, hughes2016identification, Duffaut:2020pvb}. To control and dispatch the VBs, much of the literature has focused on optimizing the VB operation over ISO/TSO market signals, but this does not directly consider the underlying AC network and the distribution system operator (DSO) constraints (e.g., voltage or power limits), see for example~\cite{tindemans2015decentralized, muller2017aggregation}.  

To avoid violating operational limits of the grid and to ensure system reliability with DERs at scale, {\color{black}coordination between a DSO and DER owners and aggregators will become critical. This has spurred a multitude of concepts and models for how DSOs can interact with DERs, aggregators, and whole-sale (transmission) markets~\cite{futureUtilities2015lbnl,kristov2016tale,conED}. In this manuscript, we focus on the so-called ``Market DSO'' model, e.g., see~\cite{futureUtilities2015lbnl}, where the DSO performs coordination and aggregation of DERs to deliver grid services. While such a setup could preclude independent DER aggregators (i.e., increases regulatory complexity), the model simplifies the interaction between wholesale market signals and the DSO and the ideas herein can be adapted further to enable independent ``grid-aware'' DER aggregators~\cite{nazir_CDC}. For other} market-based DER coordination schemes, ``transactive energy'' can engender holistic TSO-DSO-Aggregator participation of DERs~\cite{transEng}. Some of these schemes focus on broadcasting prices directly to devices. However, with large-scale participation of DERs, transactive energy is susceptible to harmful load synchronization effects, power oscillations, and volatile prices, as shown in \cite{nazir2018dynamical}. 
 \textcolor{black}{Other grid-aware approaches include optimization-based methods to account for AC network constraints}\cite{dall2018optimal}, where VB control is achieved by solving an optimization problem based on AC network models and tracking a Karush-Kuhn-Tucker (KKT) point that satisfies the KKT optimality conditions. However, for non-convex AC OPF, the KKT conditions may not be sufficient to guarantee global optimality. Other optimization schemes can provide market services with VBs without exact grid models nor real-time measurements~\cite{arnold2018model}. However, these methods do not directly incorporate multi-period energy constraints and the KKT point can be sensitive to exogenous disturbances. 
 
 Formulations based on convex relaxations can provide guarantees on feasibility and optimality of solutions when the cost function is monotonic~\cite{li2018convex}. However, optimal tracking of a power reference signal has a non-monotonic cost function, which means that one cannot guarantee that the predicted solutions are network-admissible. 
\begin{figure}[t] 
\centering
\includegraphics[width=0.97\columnwidth]{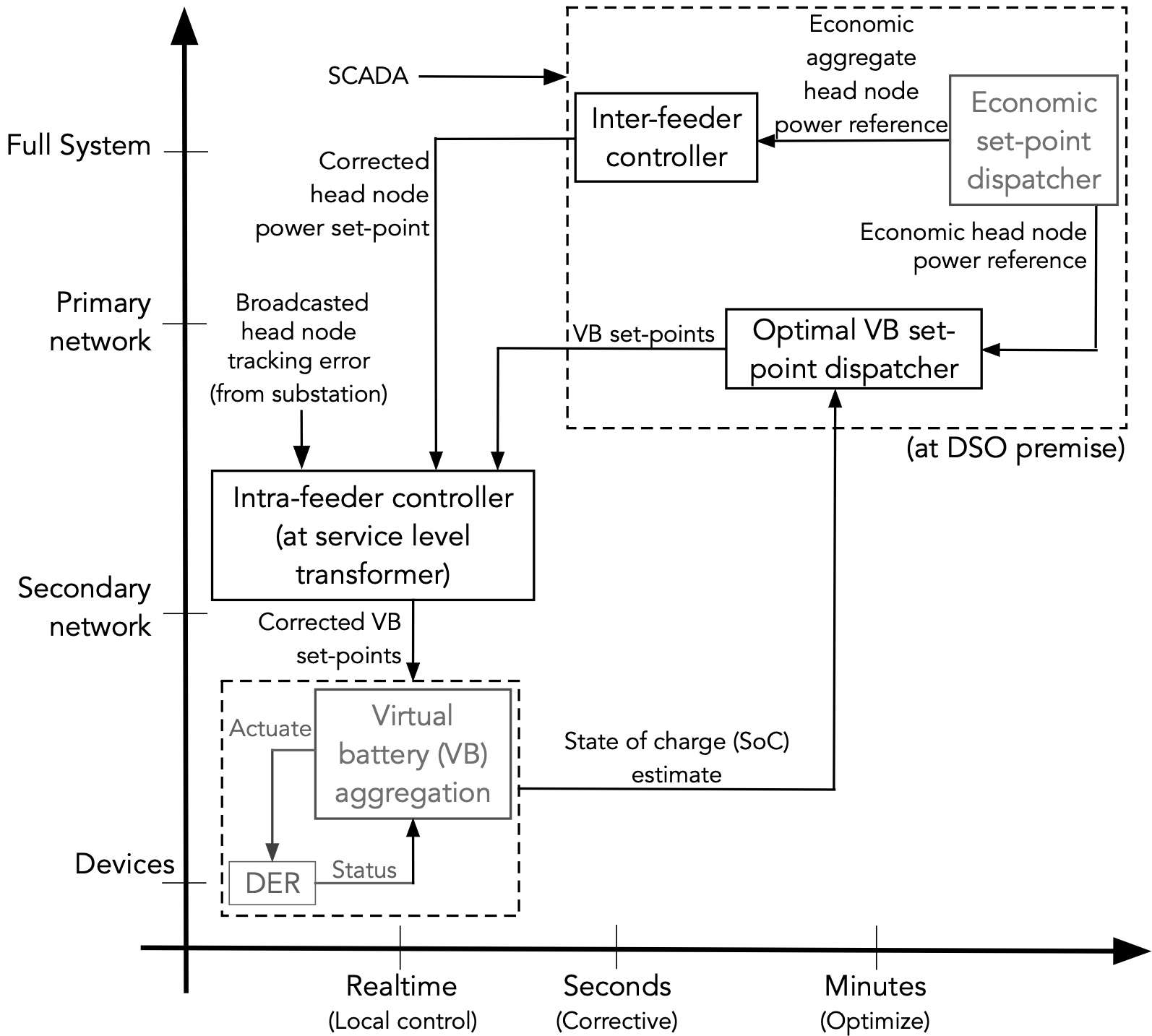}
\caption{The multi-layer VB coordination framework. The blocks with the inter- and intra-feeder controllers and optimal VB dispatcher form the focus of this paper. \textcolor{black}{The DSO runs an AC OPF about every minute to dispatch VBs with optimal set-points. The inter-feeder and the intra-feeder control action, which take place in real-time \textcolor{black}{(about every 5 s and 100 ms respectively)}, are also executed by the DSO. Specifically, the intra-feeder control computation occurs at the feeder's service transformers monitored by the DSO (this requires broadcasting from the substation to the service transformers).}}
\vspace{-15pt}
\label{fig:block_model}
\end{figure}
To address these challenges concerning real-time, optimal, and network-admissible coordination of VBs, this paper presents 
a hierarchical (multi-layer) framework for coordinating demand-side flexibility in the form of VBs (please see  Fig.~\ref{fig:block_model}).  The hierarchical coordination consists of a novel, convex OPF relaxation, which is provably tight at optimality under realistic conditions and generates grid-aware, feasible set-points for the VBs. The optimization is represented by the box ``Optimal VB set-point dispatcher" in Fig.~\ref{fig:block_model}. 

To ensure real-time tracking of VB set-points despite forecast and modeling errors and to ensure resilience, the hierarchy is augmented with a real-time controller to correct the feeder economic set-points. The real-time controller borrows concepts from wide-area control (WAC)~\cite{kundur1994power}, including (local) droop~\cite{brahma2018stochastic} and (regional) automatic generation control, and adapts them to dynamically managing VB power in distribution feeders. 
This results in a VB coordination scheme that enables a feeder head-node to optimally track a power reference while correcting, in real-time, for unexpected small (local, intra-feeder) and large (regional, inter-feeder) disturbances.

\textcolor{black}{This scheme thus effectively deals with stochasticity, both in the input disturbances to these VBs (e.g., solar variations) - through the PI control mechanism, which rejects these disturbances, as well as that in the VB parameters themselves - by re-solving the OPF problem about every minute, and re-tuning the corrective controller gains about every 5 minutes with updated parameters.} Unlike most of the existing literature, the presented framework explicitly and effectively accounts for AC (radial) networks and energy-constrained VB resources and is robust and resilient to grid disturbances. 
Similar to~\cite{hughes2016identification,tclvb},  we consider VBs at the scale of 100-200 DERs per VB that are managed locally (e.g., from the same neighborhood) via load control (with full state information) and subject to lags from device dispatch and communications. 

\textcolor{black}{Prior work on hierarchical control of DERs in microgrids (e.g. \cite{bidram2012hierarchical}) has mainly considered  using frequency and voltage droop characteristics to generate active and reactive power set-points for DERs using local measurements of frequency and voltage and compensating for the deviations, but in this paper, we compensate for the deviation in the head node power of the feeder from the economic set-point, thus taking into account the economic trajectory. Moreover, in this paper, the design of the proportional and PI gains is done based on the grid topology and device constraints, unlike local droop-based control strategies that are generally not cognizant of the network. While \cite{baker} develops a local (proportional) controller that incorporates network conditions into gains to minimize voltage deviations with active power injections, it does not consider system-wide power tracking objectives such as an economic trajectory that satisfies voltage limits across the feeder. Reference \cite{dorfler} suggests a Distributed Averaging PI (DAPI) control strategy to ensure proportional power-sharing and economic optimality, but since it requires extensive communication between the DERs, it may not be feasible on a large scale. Moreover, while the droop coefficients in \cite{dorfler} are chosen proportionally according to DER power limits, state-of-charge limits are not considered, and the coefficients are not optimized to minimize head node power deviations from the economic trajectory. In this paper, we overcome the above limitations through a novel hierarchical control of DERs.}


To be more specific, the ``optimal VB set-point dispatcher" in Fig.~ \ref{fig:block_model} considers a convex relaxation of the (balanced) AC network model.
Furthermore, the convex OPF formulation is proven tight at optimality, which guarantees that the prediction of future physical operating states of the grid and the VBs are accurate. This is achieved by decomposing the feeder head-node (i.e., substation) economic reference into the aggregate VB dispatch, net-demand, and approximated total feeder line losses. 
The key feature here is the use of a first-order prediction model for the line losses to simplify the problem. It is shown through analysis and with simulations that this VB-plus-losses reference can achieve a tight optimal solution under practical conditions. To provide the formulation with updated estimates of feeder losses and network topology, it is assumed that
a distribution system state estimator (DSSE) is available~\cite{dehghanpour2018survey} (not considered in this work).
 
The VB set-points from the above optimization problem are further augmented with a real-time controller, which enables tracking under modeling and net-load variations due to, for example, solar PV injections. The controller is divided into intra-feeder and inter-feeder operations.  The intra-feeder, droop-like proportional controller is implemented to correct the optimal VB set-points in real-time to track the feeder's head-node power reference signal. The inter-feeder proportional-integral (PI) controller with anti-windup allows for real-time disturbance rejection in the event of larger disturbance across multiple feeders, such as malfunctions in the VB's cyber network or the network structure, thus, ensuring resilience.  
The only measurement necessary for this purpose is the head-node power from the feeders. 
\vspace{-0.4cm}
\subsection{Original Contributions}
In summary, the main contributions of this paper are as follows:
\begin{itemize}
\item A hierarchical framework adapted from wide-area control for coordinated control and optimization of VBs in radial, balanced distribution feeders, that \textcolor{black}{considers multiple spatial and temporal scales.}
\item A novel second-order conic AC OPF formulation of a multi-period optimization of VBs designed for tracking the desired power reference at the head-node of the feeder. \textcolor{black}{In this formulation, the feeder structure is taken into account and the second-order cone relaxation is proven to be exact even under significant reverse power flow and with a non-monotonic cost function, which is an improvement over the present state-of-science in radial network OPF presented in~\cite{gan2015exact}.} 
\item \textcolor{black}{A practically relevant, real-time, corrective VB control scheme for intra-feeder and inter-feeder fluctuations in net-demand and localized communication outages, that requires only the measurement of the head node power, and is \emph{dynamically self-adjusting} based on  estimates of VB power and energy.}
\item Systematic, simulation-based analysis on IEEE distribution feeders is performed, illustrating the validity of the approach.
\end{itemize}
\vspace{-0.4cm}

\subsection{Practical implementation and Data management}\label{sec:data_manage}
In this subsection, we explain how the presented framework (in Fig. \ref{fig:block_model}) can be feasibly implemented in a practical scheme. In the proposed {\color{black} DSO-centric methodology, we assume that the Market DSO acts as the coordinator and aggregator of DERs} and manages the entire scheme. The DSO may also use technology services to manage the DERs via a software interface (e.g., the VB-DER interface in Fig.~\ref{fig:block_model}), but this is to help the DSO reduce costs and manage constraints. The DSO runs an OPF every minute or so for each feeder based on the economic head node power reference. \textcolor{black}{In this regard, it is reasonable to assume that the DSO has access to SCADA data and is aware of the grid topology and receives VB state of charge (SoC) estimates from the VB-DER interface to run the AC OPF.} For the VB interface, the only information required is the corrected VB power set-point. Apart from this, the VB interface does not require any other information and it is not involved in solving an OPF.

The re-tuning of controller gains takes place by executing a nonlinear optimization problem by the DSO about every 5~minutes and only if the model parameters/grid structure change. For that purpose, the requirements are the updated grid topology, VB models, the operating point from the optimal dispatcher, and the VB power and SoC every minute, but only the head node power measurement is required in real-time (on the order of 100 ms). The VB power and SoC estimates are provided by the VB-DER interface. \textcolor{black}{While every other network communication can be accomplished through SCADA, broadcasting head-node power measurements in real-time is not possible through SCADA, which operates on the order of 2-5 seconds. However, real-time automation controllers (e.g., SEL RTAC~\cite{SEL_RTAC2020}) are sufficient and can process distribution substation power measurements and broadcast the tracking error to the feeder's VBs via the DSO's Gigabit Ethernet or fiber-optic networks on a timescale of about 100 ms~\cite{llnl2019derms}.}

\section{Convex OPF formulation for reference tracking}\label{sec:convex_form}

\subsection{Mathematical Notation}\label{sec:math_not}
\textcolor{black}{Consider a radial, balanced distribution network as a graph $\mathcal{G}=\{\mathcal{N},\mathcal{E}\}$, where $\mathcal{N}$
is the set of nodes and $\mathcal{E}$ is the set of
branches, such that $(i,j)\in \mathcal{E}$, if nodes $i,j\in \mathcal{N}$ are connected, and  $|\mathcal{E}|=n, \ |\mathcal{N}|=n+1$. Node $0$ is assumed to be the head-node (i.e., substation) node with a fixed voltage $V_0$ and define $\mathcal{N^+}:=\mathcal{N}\setminus \{0\}$. 
Let $V_i$ be the voltage phasor at node $i$ and $I_{ij}$ the current phasor in branch $(i,j)\in \mathcal{E}$. Then, we define $v_i:=|V_i|^2$ and $l_{ij}:=|I_{ij}|^2$. Let $S_{ij}=P_{ij}+\mathbf{j}Q_{ij}$ denote the sending end power flow from bus $i$ to bus $j$ where $P_{ij}$ and $Q_{ij}$ denote the active and reactive power flows respectively and let $s_i=p_i+\mathbf{j}q_i$ denote the power injection into bus $i$ where $p_i$ and $q_i$ denote the active and reactive power injections, respectively. Next, denote $r_{ij}$ and $x_{ij}$ as the resistance and reactance of the branch $(i,j)\in \mathcal{E}$, respectively, which gives complex branch impedance $z_{ij}=r_{ij}+\mathbf{j}x_{ij}$.}

\textcolor{black}{Then, based on the \textit{DistFlow} model for radial networks~\cite{baran1989optimal}, the variables $(s,S,v,l,s_0)$ at any time-step $k$ 
are described by the following equations:}
  
    \begin{align}
        S_{ij}[k]=&s_i[k]+\sum_{h:h \rightarrow i}(S_{hi}[k]-z_{hi}l_{hi}[k]), \quad \forall (i,j)\in \mathcal{E}\label{eq:1}\\
        0=&s_0[k]+\sum_{h:h \rightarrow 0}(S_{h0}[k]-z_{h0}l_{h0}[k])\label{eq:2}\\
        v_i[k]-v_j[k]=&2\textrm{Re}(\overline{z}_{ij}S_{ij}[k])-|z_{ij}|^2l_{ij}[k], \quad  \forall (i,j) \in \mathcal{E}\label{eq:3}\\
        l_{ij}[k]=&\frac{|S_{ij}[k]|^2}{v_i[k]}, \quad \forall (i,j) \in \mathcal{E}\label{eq:4}
    \end{align}
    \textcolor{black}{Apart from the nonlinear relation \eqref{eq:4} of $l$ to $S$ and $v$,~\eqref{eq:1}-\eqref{eq:3} represent a linear relationship between the nodal power injections $s$, the branch power flows $S$, and the nodal voltages $v$. 
Thus, in an AC OPF optimization formulation,~\eqref{eq:4} would be a non-convex equality constraint, which begets a non-convex formulation. To overcome the non-convex formulation, we present an approach in Section~II-C that accounts for the nonlinearity introduced by~\eqref{eq:4} via a second-order cone relaxation~\cite{low2014part1,low2014part2}. However, first, we introduce the VB concept and its corresponding dynamic model, which augments the AC network constraints in the OPF formulation.}

\subsection{Virtual battery model}\label{sec:VB_model}
\textcolor{black}{In this paper, a VB is considered an energy-based model of a dispatchable aggregation of a relatively small number of controllable DERs (e.g., 100-200 distributed loads like ACs, which thus have the capacity to both consume and produce energy). This representation of DERs as an aggregate VB is adequate, based on several works in the literature (e.g., \cite{hao2014aggregate,hughes2016identification,nandanoori2019identification}). Moreover, while doing this aggregation, it is important to consider the human-in-the-loop in these flexible demand resources, which manifests itself in the form of State-of-Charge (SoC) of VBs~\cite{kane2019data}, based on Quality of Service (QoS) constraints (e.g., temperature/comfort limits).} It is assumed that the DERs reside in the low-voltage secondary network, while the local VB coordination (computation and control) takes place at a nearby primary feeder node.  Owing to the small scale of aggregation, the DERs that make up a VB are controlled directly via  utility Gigabit ethernet (e.g., IEC 61850) or wireless cellular/Wi-Fi connection. From literature, it has been shown that a VB endowed with a DER control policy, such as a priority-based switching of DERs~\cite{nandanoori2018prioritized},  can be well-described by a first-order dynamic model of the VB's energy state, (e.g., see \cite{hao2014aggregate, hughes2016identification, chakraborty2018virtual}):
\begin{align}
   \dot{B_i}(t)&=-\alpha_{\text{b},i}B_i(t)-p_{\text{b},i}(t)\label{eq:batt_update}\\
    \tau_i \dot{p}_{\text{b},i}(t)&= - p_{\text{b},i}(t) +  p_{\text{in},i}(t-T_{d,i})\label{eq:batt_1order}\\
   \underline{B_i}&\le B_i(t)\le \overline{B_i}\label{eq:batt_Blimit}\\
   \underline{p_{\text{b},i}}&\le p_{\text{in},i}(t)\le \overline{p_{\text{b},i}}\label{eq:batt_plimit},
\end{align}
where at node $i\in \mathcal{N^+}$, $B_i(t)$ is the VB's state of charge (SoC), $p_{\text{b},i}(t)$ its active power output, and $p_{\text{in},i}(t)$ its desired total active power. The upper (lower) bound of the SoC is given by $\overline{B_i}$ ($\underline{B_i}$) while the VB's upper (lower) power limits are denoted $\overline{p_{\text{b},i}}$ ($\underline{p_{\text{b},i}}$). Note that~\eqref{eq:batt_update} captures the SoC dynamics with $\alpha_{\text{b},i}$ as the energy dissipation rate. The coupling between the VB's ability to change output power and the VB's control and communication of DERs and their response is generalized with~\eqref{eq:batt_1order} as a  lag-and-delay model. In that model, $T_{\text{d},i}$ is the time delay associated with communicating with the $i$th VB and $\tau_i$ is the time constant of the first-order model (similar to \cite{samivb} for example). Specifically,~\eqref{eq:batt_1order} was formed by taking note of the following facts: $i$) The DERs that compose a VB turn on/off (possibly) sequentially, and there are power electronic components present inside each VB, both of which contribute to a net lag $\tau_i$; $ii$) \textcolor{black}{There are communication delays (generally of the order of 200 ms) between the head node of the feeder and each VB ~\cite{naduvathuparambil2002communication,amini2016investigating}, and delays associated with disaggregating the control signal into device-level signals \cite{nandanoori2018prioritized}. The delays we consider in the VB model are in fact both these types of delays lumped together. The effect of time delays on stability is investigated in Section III-D.}
For this work, the battery charge and discharge efficiencies are assumed to be unity. Inclusion of non-unity battery efficiency requires binary variables to avoid simultaneous charging and discharging. A detailed description of the battery model with non-unity efficiency and analysis on avoiding simultaneous charging and discharging is provided in~\cite{almassalkhi2015model, nazir_optimal3phase}.
\textcolor{black}{This full VB model described in \eqref{eq:batt_update}-\eqref{eq:batt_plimit} is employed in simulations involving the real-time, corrective control scheme detailed in Section~III. However, for the optimal dispatch to be presented next, since we optimize the VB set-points on a minutely timescale and $\alpha_{\text{b},i}\approx 0$ in realistic settings, it is reasonable to consider a simplified, discretized VB model for set-point optimization.} This reduced, predictive model is below:
\begin{align}
    B_i[k+1] =&  B_i[k]-\Delta t p_{\text{b},i}[k] \label{eq:B_1}\\
    \underline{B_i}\le  B_i[k]\le \overline{B_i}, \quad & 
    \underline{p_{\text{b},i}}\le p_{\text{b},i}[k]\le \overline{p_{\text{b},i}}\label{eq:B_3}
\end{align}
valid for any $k$ and where $\Delta t$ is the discretization timestep.
Next, we augment the discrete-time VB model with the AC network model to formulate the feeder head-node reference power tracking OPF problem.
\subsection{Conventional convex formulation}\label{sec:conv_form}
\textcolor{black}{  Convex relaxation techniques have gained popularity recently due to the existence of global optimality guarantees for AC OPF problems~\cite{lavaei2012zero, gan2015exact}.
In this subsection, we present a traditional second-order cone programming (SOCP) relaxation of the AC power flow equations to formulate a convex, multi-period reference tracking OPF problem, (P1). Decoupling the cost function in the form: $\sum_{i\in \mathcal{N}}f_i(p_i[k])=(p_0[k]-p_0^{\text{econ}}[k])^2+(p_{\text{VB}}^{\text{econ}}[k]-\sum_{i\in \mathcal{N^+}}p_{\text{b},i}[k])^2$, the optimization problem can be expressed as:}

\begin{subequations}\label{eq:P1}
\begin{align}\label{eq:P1_obj}
\text{(P1)}\qquad \ \min_{p_{\text{b},i}[k]} & \sum_{i\in \mathcal{N}}f_i(p_i[k])\\
\text{s.t.}: & \phantom{m} \eqref{eq:1}-\eqref{eq:3},~ \eqref{eq:B_1}-\eqref{eq:B_3}\label{eq:P1_b}\\
p_i[k]&= p_{\text{b},i}[k]-P_{\text{L},i}[k]+P_{\text{S},i}[k], \quad \forall i \in \mathcal{N^+}\label{eq:P1_real_bal}\\
q_i[k]& =-Q_{\text{L},i}[k], \quad \forall i \in \mathcal{N^+}\label{eq:P1_reac_bal}\\
l_{ij}[k]& \ge \frac{|S_{ij}[k]|^2}{v_i[k]}, \quad \forall (i,j)\in \mathcal{E} \label{eq:5}\\
s_i[k]& \in \mathcal{S}_i, \quad i\in \mathcal{N^+}\label{eq:6}\\
\underline{v_i} & \le v_i[k]\le \overline{v_i}, \quad i\in \mathcal{N^+}\label{eq:7}
\end{align}
\end{subequations}
\textcolor{black}{for discrete time-step $k\in \mathcal{T}$ over a prediction horizon $\mathcal{T}:=\{0, \hdots , T-1\}$ and where the power injection $s_i$ at a node $i\in \mathcal{N^+}$ is constrained to be in a pre-specified, compact, convex set $\mathcal{S}_i\in \mathbb{C}$. The prediction horizon length $T$ is chosen subject to availability of forecasts and depending upon the dynamics of VBs~\cite{hao2014aggregate}. In our formulation, we choose $T=10$. The set $\mathcal{S}_i\in \mathbb{C}$ depends upon the VB constraints, e.g., in case of inverters this set is given by: $\{(p_i,q_i)|p_i^2+q_i^2\le \overline{S_i^2}\}$, where $\overline{S_i}$ is the apparent power limit of the inverter. In~\eqref{eq:P1}, the cost function~\eqref{eq:P1_obj} minimizes the deviation of the head-node power $p_0[k] \in \mathbb{R}$ from the economic head-node reference trajectory $p_0^{\textrm{econ}}[k] \in \mathbb{R}$, which is composed of: $i$)  desired economic aggregate VB dispatch, $p_{\text{VB}}^\text{econ}[k] \in \mathbb{R}$; $ii$) total predicted losses, $L[k] := \sum_{(i,j)\in \mathcal{E}} r_{ij} l_{ij}[k]$; and $iii$) total forecasted net-demand, $\sum_{i\in \mathcal{N^+}} (P_{\text{L},i}[k]-P_{\text{S},i}[k])$, where $P_{\textrm{L},i}[k]\in \mathbb{R}$ is the active power demand at node $i\in \mathcal{N^+}$, and $P_{\text{S},i}[k] \in \mathbb{R}_+$ is the solar PV generation at node $i\in \mathcal{N^+}$. We are optimizing over the VB dispatch, $p_{\text{b},i}[k] \in \mathbb{R} \quad \forall i \in \mathcal{N^+}$, which appears in power balance constraint~\eqref{eq:P1_real_bal} while reactive power demand, $Q_{\textrm{L},i}[k]\in \mathbb{R}$, is used in~\eqref{eq:P1_reac_bal}. The second-order cone relaxation of the nonlinear equation \eqref{eq:4} is given in \eqref{eq:5}; \eqref{eq:6} and \eqref{eq:7} provide constraints on power injection and voltage magnitudes.}

\textcolor{black}{Several works in literature such as~\cite{gan2015exact}
provide conditions under which the second-order cone relaxation is exact for distribution networks.
If an optimal solution of (P1) $w^*=(s^*,S^*,v^*,l^*,s_0^*)$ is feasible for OPF, i.e., $w^*$ satisfies \eqref{eq:4}, then $w^*$ is global optimum of OPF and (P1) is said to be exact. Theorem 1 in~\cite{gan2015exact} provides conditions for the SOCP problem in (P1) to be exact, however, it requires the part of cost function $f_0(p_0)=(p_0-p_0^{\text{econ}})^2$ to be strictly increasing, which is not the case when tracking a reference power signal. Thus, even under the conditions provided in~\cite{gan2015exact}, (P1) may not be exact. 
    To overcome these shortcomings of (P1), we propose a novel method for convexifying the AC OPF while ensuring an exact solution at optimality. Specifically, we utilize a linearized approximation of line losses and through this obtain a cost function that is strictly increasing in $p_0$ in order to satisfy the conditions in~\cite{gan2015exact}.}
\subsection{Reformulated convex formulation}\label{sec:rcf}
\textcolor{black}{To reformulate (P1), we consider each piece of the two-part composition of feeder's predicted head-node power, $p_0[k]\approx -\sum_{i\in \mathcal{N^+}} p_i[k]+L_1[k]$. Specifically, we employ a first-order approximation of total predicted line losses, $L_1[k]$, in the cost function (via $p_i$-to-loss sensitivity factors, which we denote by $\zeta_i$), and prove that with approximated losses, the solution is tight at optimality. Thus, the predicted grid response to an optimized VB dispatch is AC feasible. 
To achieve the above, we use $p_{\text{VB}}^{\text{econ}}$ and $p_0^{\text{econ}}$ mentioned earlier. This leads to a multi-objective reference tracking problem, similar to the form of a linear quadratic regulator (LQR) from optimal control~\cite{kwakernaak1972linear}:} 
\begin{subequations}\label{eq:P3}
\begin{align}
\textrm{(P2)}\quad \min_{p_{\text{b},i}[k]} & \sum_{k=1}^T f_{\text{HN}}[k]^2+\alpha f_{\text{VB}}[k]^2+ \epsilon p_0[k]\label{eq:P2_obj}\\
\textrm{subject to:} &\nonumber\\
f_{\text{HN}}[k]&= \sum_{i\in \mathcal{N^+}} (p_{\text{b},i}[k]+P_{\text{S},i}[k]-P_{\text{L},i}[k])-L_1[k]- p_0^\text{econ}[k]\label{eq:P2_HN}\\
f_{\text{VB}}[k] &=p_{\text{VB}}^{\text{econ}}[k]-  \sum_{i\in \mathcal{N^+}} p_{\text{b},i}[k]\label{eq:P2_VB}\\
L_1[k]&=L_{0,k}+\sum_{i \in \mathcal{N^+}}\zeta_i\Delta p_i[k]\label{eq:P2_loss}\\
&\textrm{\eqref{eq:1}-\eqref{eq:3}, \eqref{eq:B_1}, \eqref{eq:B_3}, \text{ and }\eqref{eq:P1_real_bal}-\eqref{eq:7}}\label{eq:P2_final},
\end{align}
\end{subequations}
\textcolor{black}{for discrete time-step $k\in \mathcal{T}$ over a prediction horizon $\mathcal{T}:=\{0, \hdots , T-1\}$. The parameters $\alpha$,
$\epsilon$ $\in \mathbb{R}_+$ are chosen appropriately with $\epsilon\ll 1$; $L_1[k] \in \mathbb{R}$ is the first-order estimate of the total feeder line losses at time-step $k$ and $L_{0,k} \in \mathbb{R}_+$ is the loss estimated for the operating point at time $k$.
The term $\epsilon p_0$ results in a tight relaxation as will be shown next in Theorem \ref{thm:Th1}, whereas the term $\zeta_i\Delta p_i[k]=\zeta_i(p_i[k]-p_i[0])$ represents the change in network loss due to change in active power injection at node $i \in \mathcal{N^+}$, with $p_i[0]$ being the nominal injection.
The factors $\zeta_i \in \mathbb{R}$ provide the first-order change in feeder losses due to changes in VB power injections. Similar \textit{power transfer distribution factors} are often used in transmission system analysis but have recently been adapted for distribution networks~\cite{Christakousensitivity}.} 

\begin{remark}
The formulation in (P2) can easily be extended to account for solar curtailment as a control variable resulting in a more general formulation. If $P_{\text{C}}^{\text{econ}} \in \mathbb{R}_+$ represents the curtailment reference trajectory and $P_{\text{C},i} \in \mathbb{R}_+$ represents the solar curtailment at node $i\in \mathcal{N^+}$, then $\sum_{i\in \mathcal{N^+}}(P_{\text{S},i}-P_{\text{C},i})$ is the net solar output and $\sum_{i\in \mathcal{N^+}} P_{\text{C},i}-P_{\text{C}}^{\text{econ}}$ is the error in tracking the curtailment trajectory.
\end{remark}
In the next section, Theorem~\ref{thm:Th1} proves that under practical conditions, the (P2) has a zero duality gap.
\subsection{Exactness of formulation (P2)}

\textcolor{black}{In order to explain the notation in Theorem~\ref{thm:Th1}, consider $\mathcal{L}:=\{l\in \mathcal{N}|  \not \exists k\in \mathcal{N}\  \text{such that}\  k \rightarrow l\}$, which denotes the collection of leaf nodes in the network. For a leaf node $l\in \mathcal{L}$, let $n_l+1$ denote the number of nodes on path $\mathcal{P}_l$, and suppose}
    \begin{align}
       \mathcal{P}_l=\{l_{n_l}\rightarrow l_{n_{l-1}}\rightarrow \hdots l_1\rightarrow l_0\} \nonumber
    \end{align}
   \textcolor{black}{ with $l_{n_l}=l$ and $l_0=0$.} 
    
    \textcolor{black}{Also, define $a^+:=\mathrm{max}\{a,0\}$ for $a\in \mathbb{R}$ and let $I_2$ denote the $2\times 2$ identity matrix, and define vectors $u_i := \text{col}\{r_{ij}, x_{ij}\}$ and matrices 
    $$\underline{A}_i:=I_2-\frac{2}{\underline{v_i}}\begin{bmatrix}r_{ij}\\x_{ij}\end{bmatrix}\begin{bmatrix}\hat{P}_{ij}^+(\overline{p})\ \ \hat{Q}_{ij}^+(\overline{q})\end{bmatrix}$$ for $(i,j)\in \mathcal{E}$
    where $\hat{P}_{ij}^+(\overline{p})$ and $\hat{Q}_{ij}^+(\overline{q})$ are upper bounds on $P_{ij}$ and $Q_{ij}$ and are chosen so that $\underline{A}_i$ only depends on the SOCP parameters $(r,x,\overline{p},\overline{q},\underline{v})$.}
   
    \textcolor{black}{Furthermore, let $(\hat{S},\hat{v})$ denote the solution of the Linear DistFlow model, then}
    \begin{align}
        \hat{S}_{ij}(s)=&\sum_{h:i\in \mathcal{P}_h}s_h, \quad \forall (i,j)\in \mathcal{E}\\
        \hat{v_i}(s):=&v_0+2\sum_{(j,k)\in \mathcal{P}_i}\mathrm{Re}(\overline{z}_{jk}\hat{S}_{jk}(s)), \quad \forall i\in \mathcal{N^+}
    \end{align}
   \textcolor{black}{where $\mathcal{P}_i$ denotes the unique path from node $i$ to node 0. Since the network is radial, the path $\mathcal{P}_i$ exists and is unique. Physically, $\hat{S}_{ij}(s)$ denotes the sum of power injections $s_h$ towards node 0 that go through line $(i,j)$. Note that $(\hat{S}(s),\hat{v}(s))$ is affine in $s$, and equals $(S,v)$ if and only if line loss $z_{ij}l_{ij}$ is 0 for $(i,j)\in \mathcal{E}$. Then based on the DistFlow model define:}
    \begin{align}
        \mathcal{S}_{\text{volt}}:=\{s\in \mathbb{C}^n|\hat{v_i}(s)\le \overline{v_i}\quad \forall i\in \mathcal{N^+}\}
    \end{align}
    \textcolor{black}{which denotes the power injection region where $\hat{v}(s)$ is upper bounded by $\overline{v}$. Since $v(s)\le \hat{v}(s)$ (Lemma 1 in~\cite{gan2015exact}), the set $\mathcal{S}_{\text{volt}}$ is a power injection region where voltage upper bounds do not bind. Then based on this notation, Theorem~\ref{thm:Th1} below proves the exactness of (P2).}

\textcolor{black}{
\begin{theorem}\label{thm:Th1}
  The SOCP problem (P2) is exact if the C1 and C2 conditions given in Theorem~1 of~\cite{gan2015exact} are satisfied:
    \begin{itemize}
        \item[\textbf{C1}:]  $\underline{A}_{l_s}\underline{A}_{l_{s+1}}\hdots \underline{A}_{l_{t-1}}u_{l_t}>0$ for any $l\in \mathcal{L}$ and any $s,t$ such that $1\le s\le t\le n_l$;
        \item[\textbf{C2}:]  Every optimal solution $w^*=(s^*,S^*,v^*,l^*,s_0^*)$ satisfies $s^*\in \mathcal{S}_{\text{volt}}$
    \end{itemize}
\end{theorem}
\begin{proof}
 The cost function of the optimization problem (P2) can be expressed as:
    \begin{align}
        f_0(p_0)=&\epsilon p_0\\
        f_i(p_i)=&f_{\text{HN}}^2+\alpha f_{\text{VB}}^2 \qquad \forall i\in \mathcal{N^+}
    \end{align}
    As $f_0$ in the cost function in~\eqref{eq:P2_obj} is strictly increasing, the SOCP formulation satisfies all the conditions provided in Theorem 1 of~\cite{gan2015exact} and hence the proof is a direct application of Theorem 1 in~\cite{gan2015exact} under conditions C1 and C2. This concludes the proof. 
\end{proof}
}
\textcolor{black}{The term $\epsilon p_0$ is added to satisfy the additional condition in Theorem 1 of \cite{gan2015exact}, where the cost function must be increasing with respect to $p_0$. Note that the inclusion of this term in the cost function affects the optimal solution. However, $\epsilon > 0$ can now be made arbitrarily small (per the proof of Theorem~\ref{thm:Th1}), which ensures that the impact on the optimal solution is negligible.}

 \textcolor{black}{C1 can be checked apriori and efficiently since $\underline{A}$ and $u$ are simple function of $(r,x,\overline{p},\overline{q},\underline{v})$ that can be computed in $\mathcal{O}(n)$ time and there are no more than $n(n+1)/2$ inequalities in C1.
    For practical parameters ranges of $(r,x,\overline{p},\overline{q},\underline{v})$, line resistance and reactance $r_{ij},x_{ij}<<1$ per unit for $(i,j)\in \mathcal{E}$, line flow $\hat{P}_{ij}(\overline{p}), \hat{Q}_{ij}(\overline{q})$ are on the order of 1 per unit for $(i,j)\in \mathcal{E}$ and voltage lower bound $\underline{v}_i\approx 1$ per unit for $i\in \mathcal{N^+}$. Hence, $\underline{A}_i$ is close to $I$ for $i\in \mathcal{N^+}$, and therefore C1 is likely to hold. As has been shown in~\cite{gan2015exact}, C1 holds for several test networks, including those with high penetration of renewables.}
    
     \textcolor{black}{To show the practical restriction of condition C2, Fig.~\ref{fig:C2_flex} shows the increase in $\hat{v}$ with increase in reverse power flow due to increased VB injections. From the figure, it can be seen that the condition is valid for VB injections up to more than 400\% of demand, compared to the base case of 20\%. It can also be seen from Fig.~\ref{fig:C2_flex} that $\hat{v}$ matches the actual voltage $v$ very closely due to the low impedance of the IEEE-37 node system resulting in small loss term $z_{ij}l_{ij}$. However, with solar PV penetration and an increase in impedance values, the maximum VB injection limit will reduce.}

     The next section presents simulation results that illustrate the effectiveness of (P2).
    
\begin{figure}
\centering
\includegraphics[width=0.7\columnwidth]{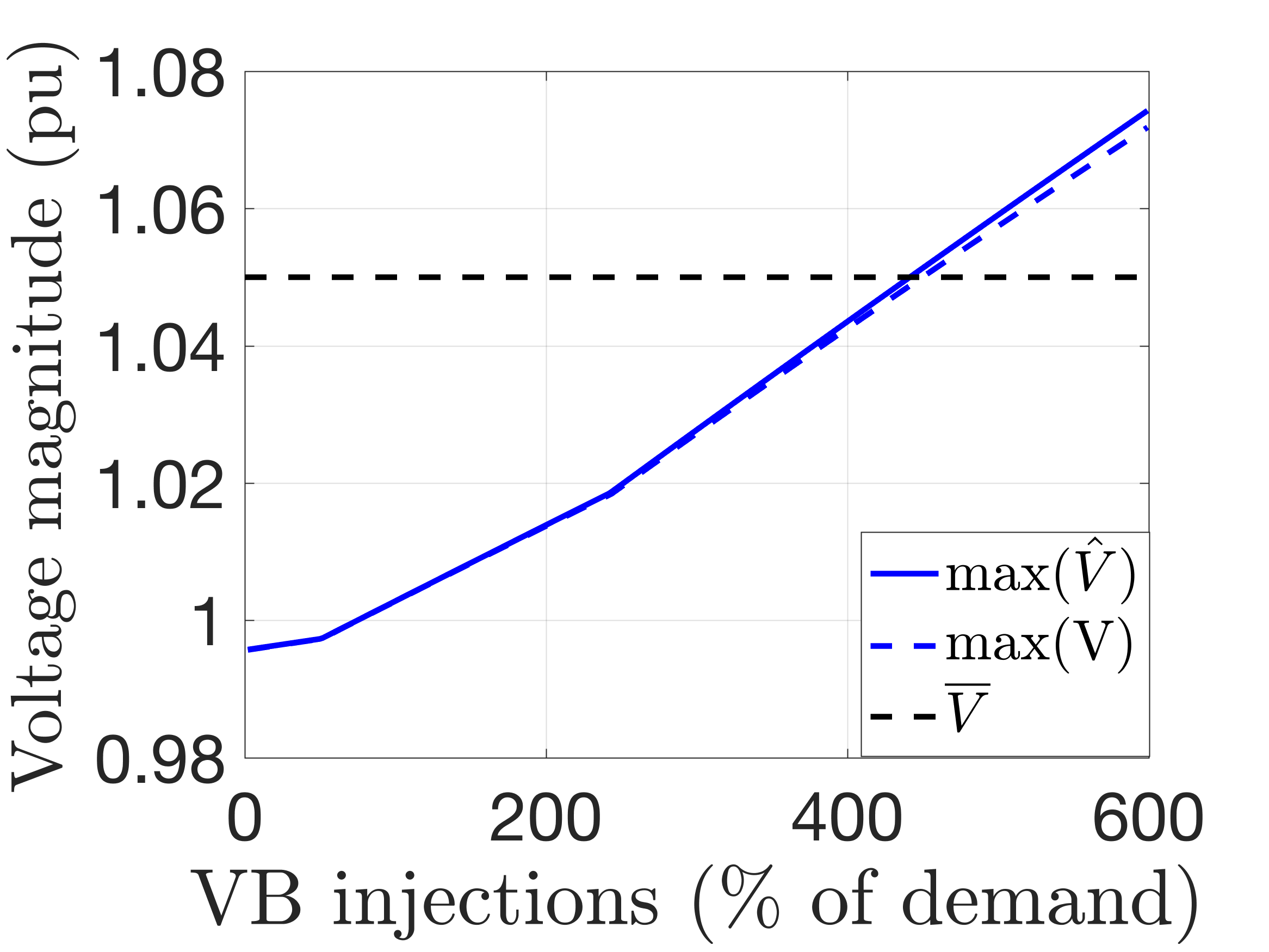}
\vspace{-10pt}
\caption{Condition C2 holds under large reverse power flow from VB injections up to around 400\% of demand.}
\label{fig:C2_flex}
\vspace{-15pt}
\end{figure}
    
    
\subsection{Optimal VB dispatch and head node tracking simulation}\label{sim_results}
Simulation tests for the optimal reference tracking with VBs were conducted on the balanced version of the IEEE-37 node test feeder~\cite{kersting2001radial} to compare the conventional (P1) and the proposed (P2) formulations against the actual AC load flow from Matpower~\cite{zimmerman2011matpower}. Simulation results in Fig.~\ref{fig:head_tracking} show the reference tracking results in Matpower achieved through the flexibility of VBs using (P1) and (P2). The figure shows that (P2) can track the reference trajectory whereas (P1) cannot. For (P2), the error in tracking at each step change in the reference trajectory is due to the first-order loss approximation used in (P2). Since the loss approximation is updated every time-step, the effect on tracking error is small and corrected quickly as shown in Fig.~\ref{fig:head_tracking}, which means that (P2) represents a reference-tracking OPF formulation that can effectively dispatch VBs while guaranteeing network admissibility. On the other hand, for the convex (P1), the non-zero duality gap creates a mismatch between the predicted power flow values and the actual AC power flow, which results in the sub-par tracking illustrated in Fig.~\ref{fig:head_tracking}. Specifically, (P1) predicts perfect tracking, but the realized AC head node power does not match the grid reference, which results in suboptimal use of VB resources. The comparison of the aggregate State of Charge (SoC) obtained through (P1) and (P2) is shown in Fig.~\ref{fig:SoC_comp}. Clearly, (P1) predicts a different SoC trajectory than (P2) due to the non-physical solution of (P1). 
\begin{figure}[t]
    \vspace{-9pt}
  \subfloat [\label{fig:head_tracking}]{   \includegraphics[width=0.485\linewidth,trim={0 0 0 0.5cm},clip]{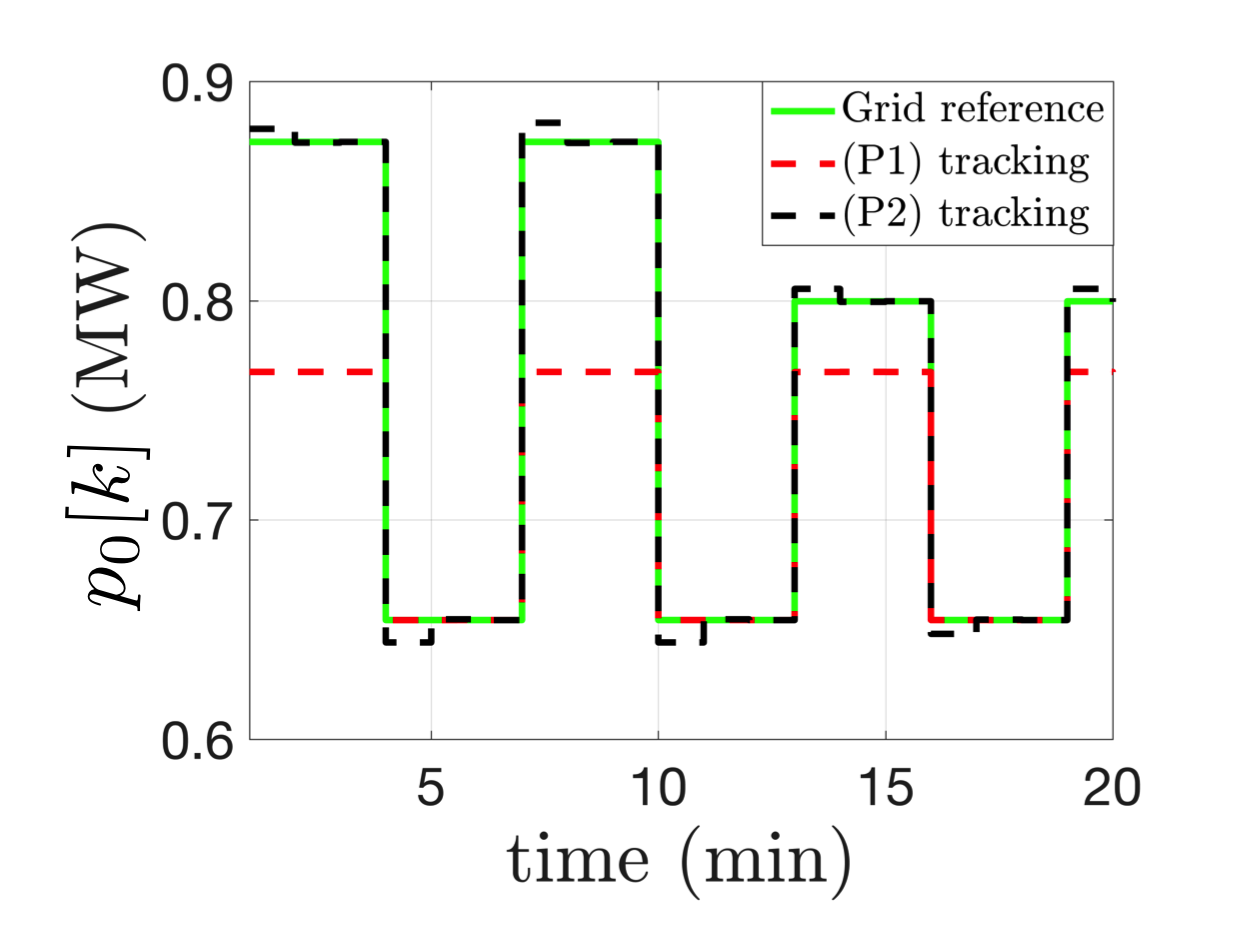}}
    \hfill
  \subfloat [\label{fig:SoC_comp}]{    \includegraphics[width=0.485\linewidth,trim={0 0 0 0.5cm},clip]{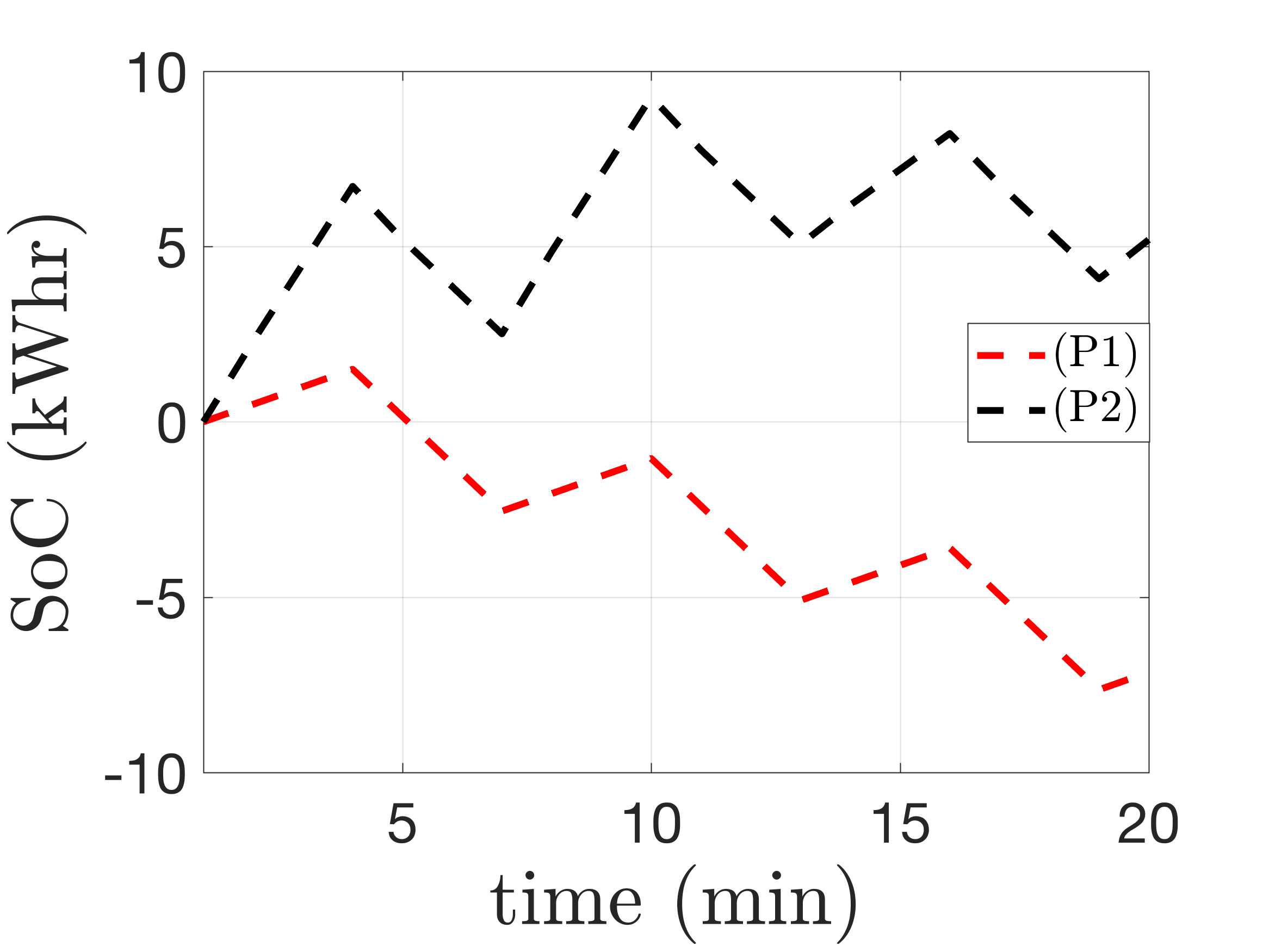} }
  \caption{Comparing the tracking performance of (\textit{P1}) vs (\textit{P2}); (b) Comparison of predicted SoC between (\textit{P1}) and  (\textit{P2}) over the prediction horizon.}
  \label{fig:tracking}
  \vspace{-10pt}
\end{figure}

\begin{remark}[High voltage conditions]
{\color{black} The voltage condition in Theorem~\ref{thm:Th1} (C2) is not restrictive for practical distribution networks as can be seen from Fig.~\ref{fig:C2_flex}. Importantly, condition (C2) can still be satisfied under large reverse power flows, which occurs in feeders with significant penetrations of batteries or solar~PV generation. To illustrate the effect of reverse power flows, we present simulation results in Fig.~\ref{fig:sim1}. For example, in Figs.~\ref{fig:power_feas} and~~\ref{fig:volt_feas}, it is shown that the predicted active head-node power matches the actual power while satisfying the voltage constraints. Of course, reverse power flows and high voltage conditions are related, which means that we are assuming that appropriate DER hosting capacity studies have been conducted to inform operations and avoid high voltage conditions. Nonetheless, theoretically, there are reverse power flows for which condition (C2) is violated at optimality, which means that the convex relaxation in (P2) may not be tight. Figs.~\ref{fig:power_infeas} and~\ref{fig:volt_infeas} illustrate the effects of a non-tight solution and show that the mismatch between the predicted and actual head-node power can lead to voltage violations due to predicted (relaxed, fictitious) losses that ensure a feasible solution in (P2), but are not realized in the physical feeder and, thus, reduce the head-node power further. To guarantee an exact solution that is always physically meaningful, we could include additional constraints to (P2) that capture condition (C2) implicitly (e.g., augment (P2) with the LinDist formulation's voltage variables, $\hat v_i$, and $\hat v_i$'s upper voltage bound), which ensures that $v< \bar{v_i}$, if (P2) is feasible~\cite{gan2015exact}. Alternatively, we could just project (P2)'s optimal solution onto the AC feasible set and accept the loss of optimality. }
\end{remark}

\begin{figure}
\vspace{-12pt}
\centering
\subfloat[\label{fig:power_feas}]{\includegraphics[width=0.49\linewidth]{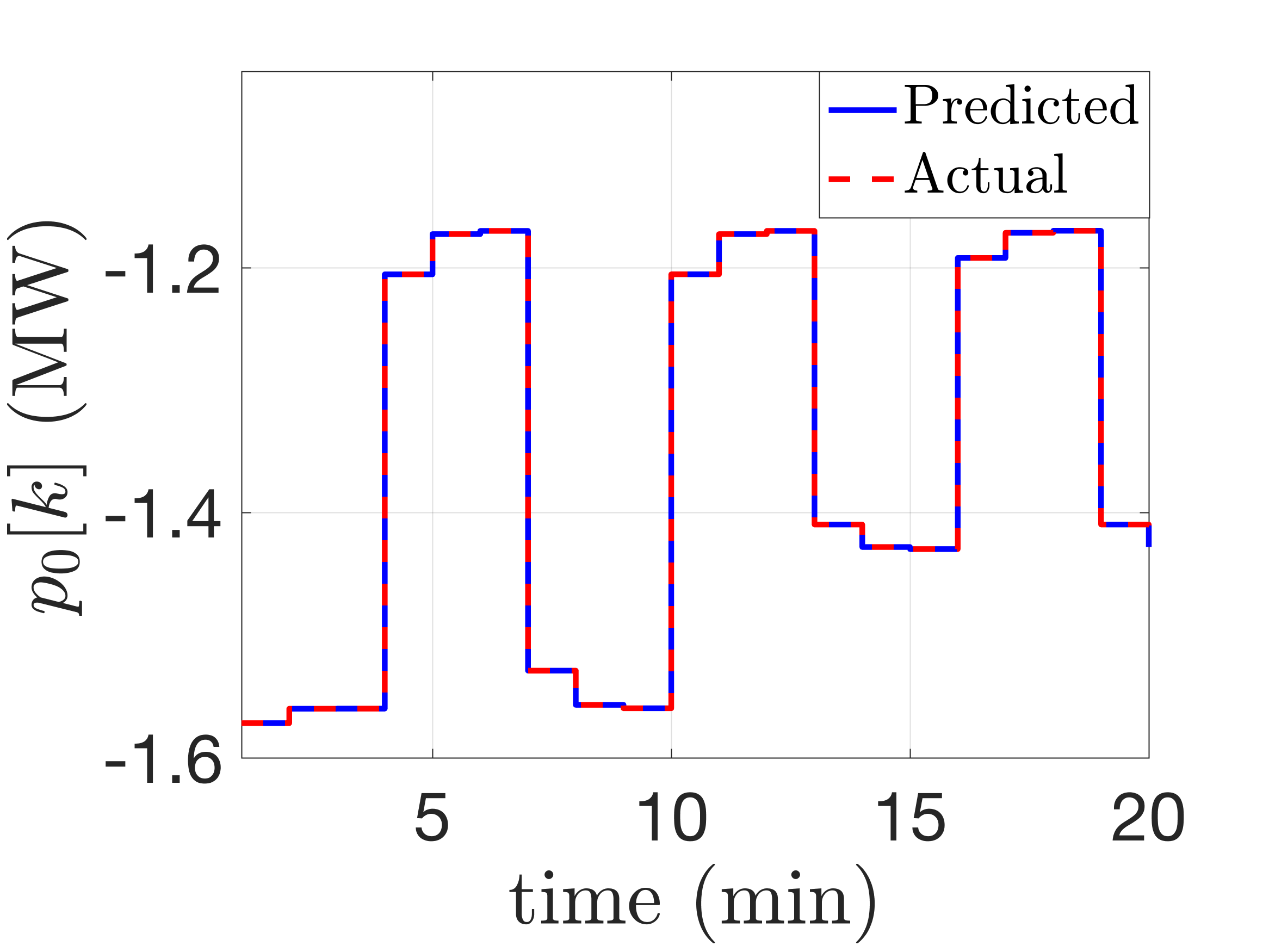}}
\hfill
\subfloat[\label{fig:volt_feas}]{\includegraphics[width=0.49\linewidth]{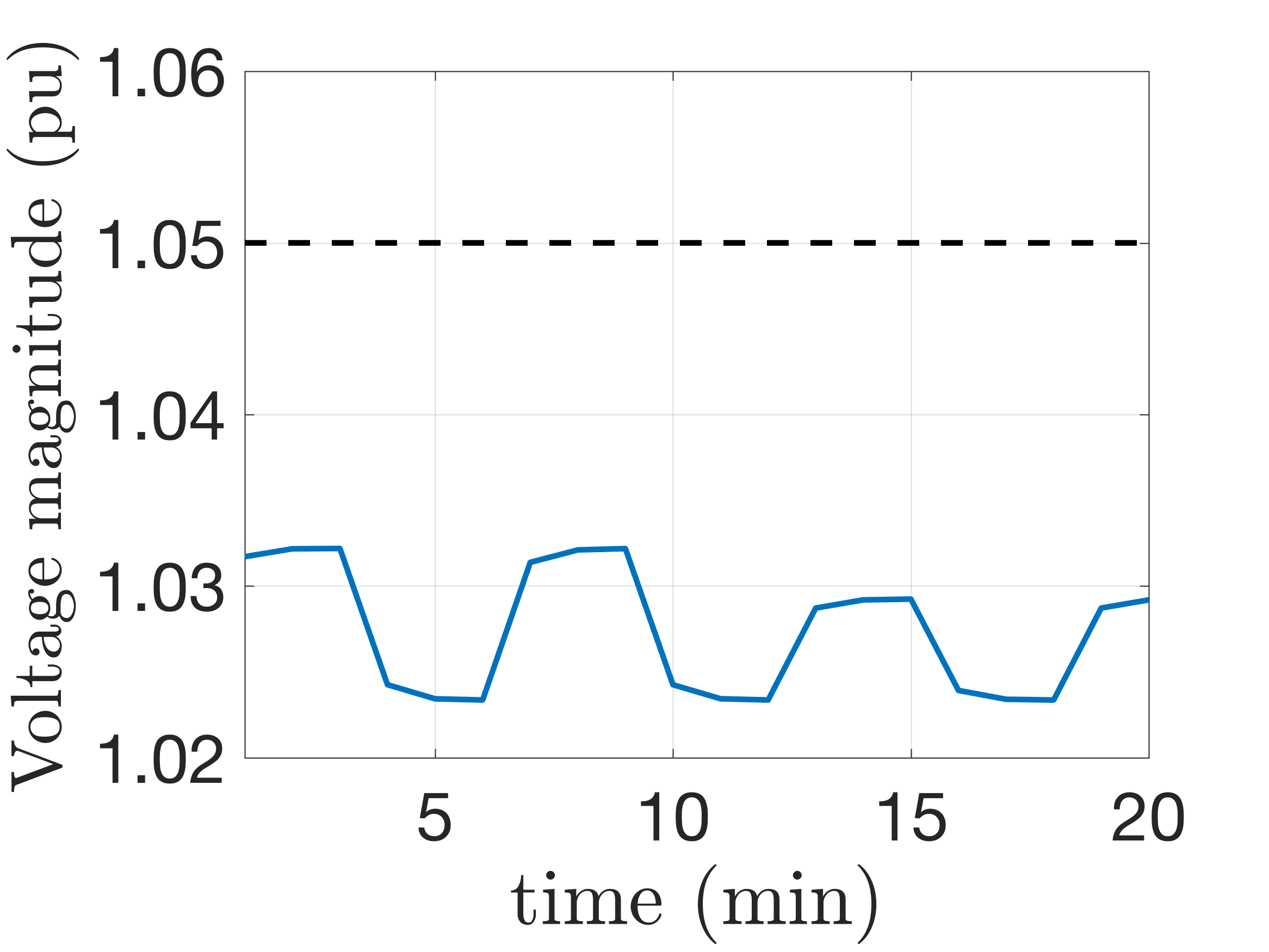}}
\\
\vspace{-12pt}
\subfloat[\label{fig:power_infeas}]{\includegraphics[width=0.49\linewidth]{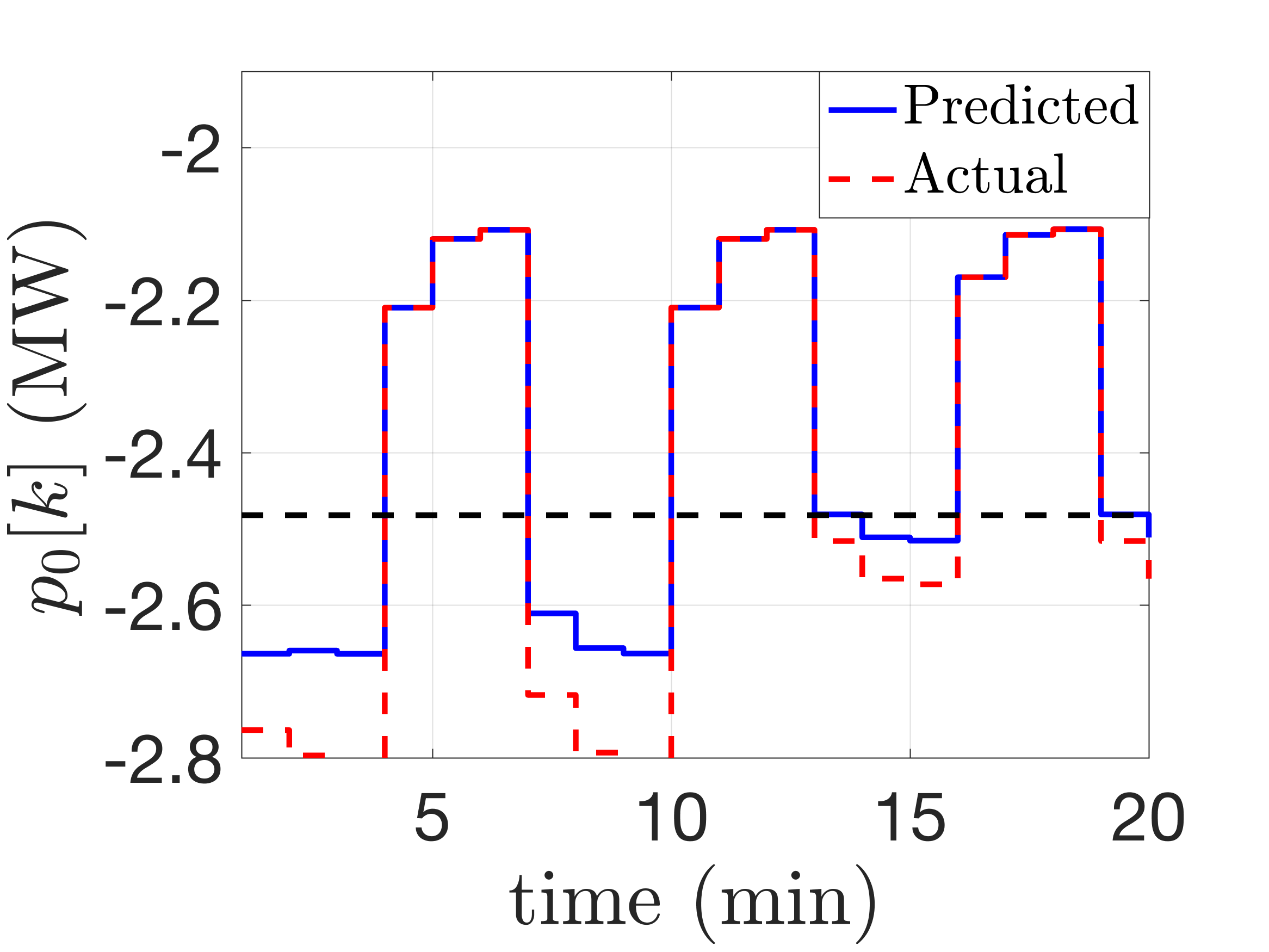}}
\hfill
\subfloat[\label{fig:volt_infeas}]{\includegraphics[width=0.49\linewidth]{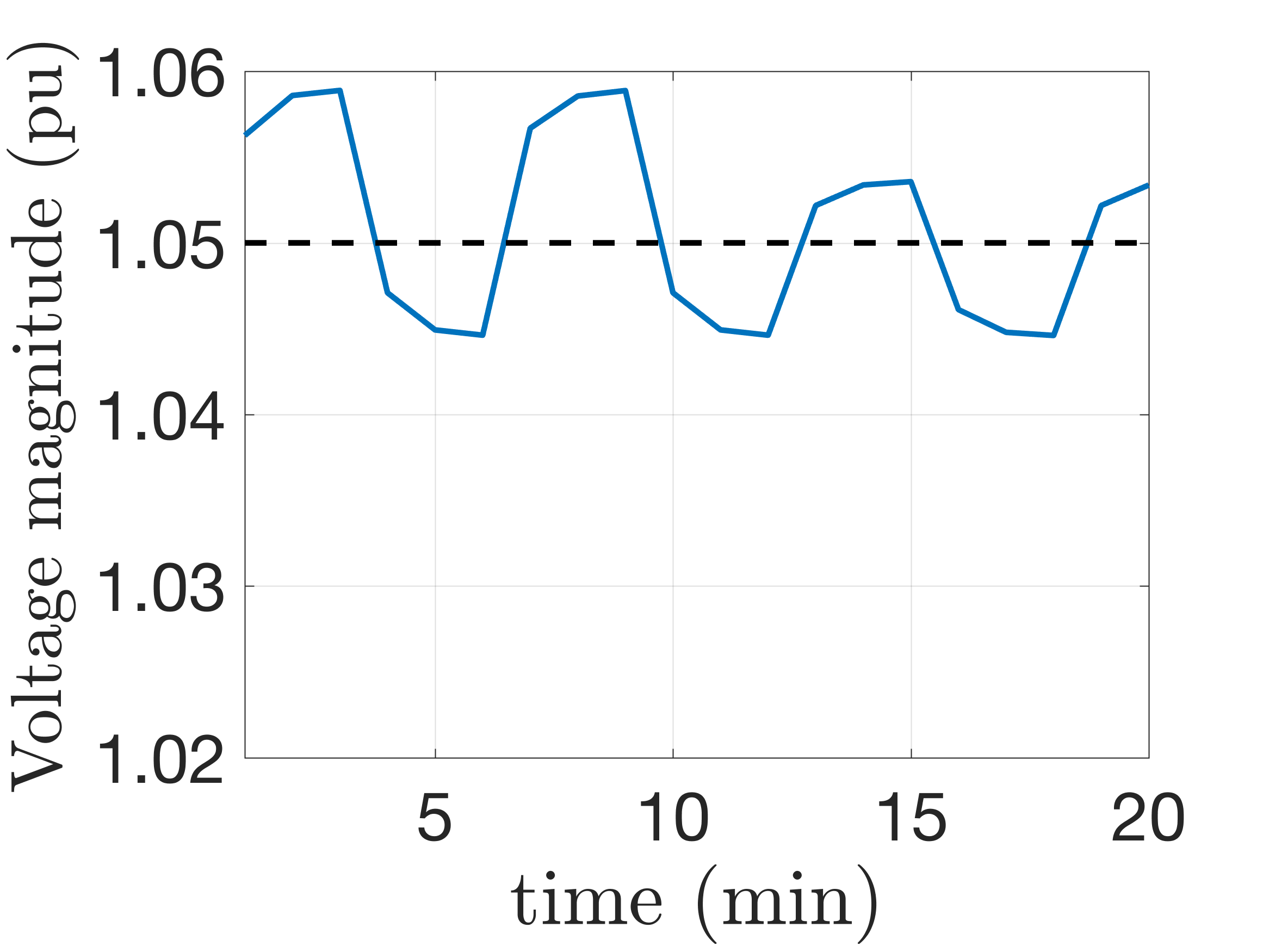}}
\vspace{-8pt}
\caption{Reverse power flow studies on (P2): (a) Comparison between predicted and actual head-node power under reverse power flow when C2 holds which leads to an exact solution, (b) Admissible voltages under reverse power flow when C2 holds, (c) Comparison between predicted and actual head-node power when C2 does not hold (C2 limit shown by black dotted line), leading to a solution that is non-exact, (d) Voltage violation due to the solution being non-exact when C2 does not hold.} \label{fig:sim1}
\end{figure}

Note that the solve time for the above optimization problem is typically less than a second for a feeder with 37 nodes. Our prior work on large scale three-phase systems has shown that the OPF scales well and can be solved in under a minute~\cite{nazir_optimal3phase}. The next section develops the real-time control of VBs to achieve tracking of grid reference by countering variations in net-demand and also to provide support under contingencies. In Section \ref{sec:sim}, the OPF is integrated with the real-time corrective control action to show the effectiveness of the combined approach through simulation results.

\section{Resilient \& Corrective control \\ of Virtual Batteries}\label{sec:rr}

The optimal VB set-point dispatcher depicted in Fig. \ref{fig:block_model} solves (P2) about every minute  and provides optimal active power set-points to all VBs in the feeder, such that the economic power trajectory for each feeder is tracked optimally. The aggregate of DERs that make up these VBs is then expected to provide and maintain those power set-points until the next set-point update. However, due to the nature of distribution feeders with solar PV, there will inevitably be short-term fluctuations in net-demand, which act as disturbances within a feeder. The flexibility available from the VBs can be used to mitigate these \textit{intra}-feeder disturbances, by \textit{correcting} the set-points provided by the optimal VB set-point dispatcher. Furthermore, for a utility with multiple feeders connected to a substation, one feeder may suffer from larger disturbances, e.g., forecast errors not accounted for in (P2), cyberattacks on the VBs' or DERs' communication channels, and changes to network topology from local outages. In these cases, it becomes important for the system to be \textit{resilient} \cite{taft2017electric} and maintain the economic set-point provided by the market despite such \textit{inter}-feeder disturbances.  In this section, we hence provide a \textit{resilient} and \textit{corrective} control mechanism for mitigating these \textit{intra}-feeder and \textit{inter}-feeder disturbances by leveraging the flexibility of VBs. This ensures that feeders with high penetration of solar PV can be effectively dispatched to provide energy market services. \textcolor{black}{This “real-time” corrective action will be executed in the order of a few hundred milliseconds.}

The ideas proposed here are similar to existing ideas in wide-area control, including primary (droop) control and automatic generation control (AGC) of frequency ~\cite{kundur1994power}, but they are adapted for distribution system operations and regulating power. 
For example, the intra-feeder control mechanism is needed to respond to small disturbances within a feeder without invoking VBs in other feeders and hence needs to operate on a fast time scale (similar to primary frequency control). The inter-feeder control, on the other hand, needs to operate only when there is a large contingency that necessitates a response from VBs present in all the feeders, and thus its response is expected to be slower than intra-feeder control, which is similar to automatic generation control (AGC). Unlike primary frequency control and AGC, we also need to take into account saturation in the energy and power of the VBs while designing for and implementing the real-time controllers. 

\subsection{Overview of Intra-Feeder Control of Virtual Batteries}\label{ss:ifc}
The purpose of the intra-feeder control scheme is to reject short-term disturbances (like solar PV/demand fluctuations) that enter the nodes inside a feeder and maintain the economic market set-point at the feeder's head-node (i.e., substation). Our proposed intra-feeder control scheme essentially consists of a bank of proportional controllers, one to control each VB in the feeder. For this purpose, the only measurement required is the active power at the head-node of the feeder, which is readily available at the substation. The block diagram in Fig.~\ref{fig:intrafeeder} represents one feeder with $n$ VBs and disturbances in the form of uncontrollable, unscheduled nodal injections. The VBs are modeled as mentioned in Section~\ref{sec:VB_model} where the power and energy bounds of the VBs are represented by the saturation blocks. $P_{uf}$ denotes the corrected economic reference head-node power for this feeder, as calculated by the inter-feeder controller described in the next section, and needs to be tracked by this intra-feeder control scheme. $P_h$ denotes the head node power of the feeder. 
\textcolor{black}{ The corrected VB set-point for the $r$th VB, ${p}_{\textrm{in},r}$, is then obtained as follows:
\[
{p}_{\textrm{in},r}=K_{r}K_{\textrm{adj},r}\left(P_{uf}-P_{h}\right)+P_{\textrm{set},r}\] where $P_{\textrm{set},r}$ refers to the set-point provided by the optimal set-point dispatcher about every minute, $K_r$ is a proportional controller that is further adjusted by a factor $K_{\textrm{adj},r}$ (as defined in \eqref{eq:adj}, see below).
The net VB power, $p_{b,r}$, (computed from \eqref{eq:batt_1order}) is injected into the ``Feeder" block (which represents the feeder with nodal active power injections as input and its head node active power as output) at the respective location. The ``Feeder" block is linearized when designing the gains, as will be described later.} The plant that is controlled here thus consists of the VBs, together with the grid. 
        Estimates of the SoC signal, $B_r$, computed by the DSO based on DER models, and the net VB power signals, $p_{b,r}$, obtained from the VB interfaces, are used to dynamically modify the adjustment factor for the proportional gain (described below), $K_{\textrm{adj},r}$ about every minute (this adaptive and dependent behavior of $K_{\textrm{adj},r}$ is indicated by slanted arrows in Fig. \ref{fig:intrafeeder}). This adjustment factor ensures that an unreasonable value of power that might deplete or fully charge the VB is not demanded from the VB. The computation of the control input is performed at the nodal service transformers that are monitored by the DSO. To obtain the response from the set of DERs in practice, however, we are assuming that a device level algorithm (like the priority-based scheme in [30]) is in place to ensure that the DERs can provide the response that is desired from them. This device-level dispatch can be performed every few hundreds of milliseconds.

\begin{figure}
\centering
\includegraphics[width=1\columnwidth]{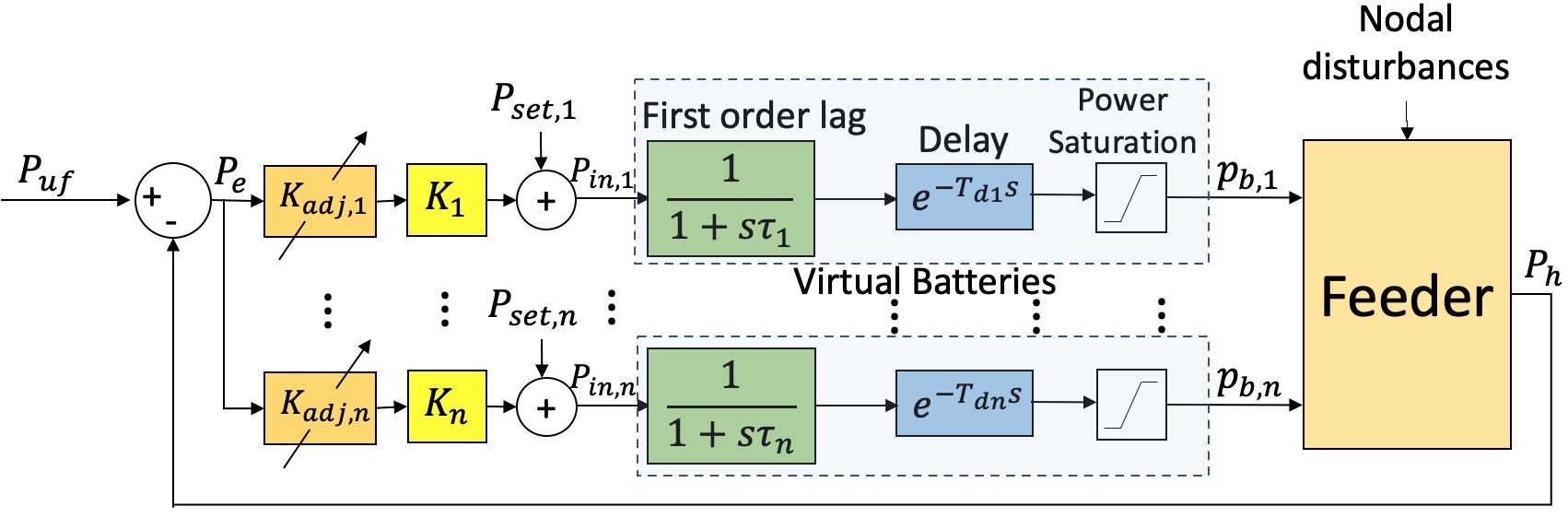}
\vspace{-10pt}
\caption{Intra-feeder Control.  \textcolor{black}{The adjustment factors $K_{adj,r}$ are adaptive, and depend on the VB power, $p_{b,r}$ and the SoC estimate, $B_r$, (obtained through the VB interface) as given by \eqref{eq:adj}. Note that the signal $P_{uf}$ is updated about every 5 minutes and $P_{set,r}$ are updated about every minute, thus operating at a slower time-scale compared to the other signals.}}
\label{fig:intrafeeder}
\end{figure}

To ensure a good disturbance rejection capability, we select the proportional gains, $K_r$, optimally, using an approach similar to standard LQR.  Specifically, we assume that the system (i.e., the feeder) is affected by nodal injections that are stochastic. This is reasonable since solar PV and demand fluctuations are typically random. The objective is then to reduce the effect of these injections on deviating the head-node power from the economic reference. Alternatively, we can treat the reference signal itself as stochastic and reduce the standard deviation of the tracking error. To do that, first, the system is linearized by 1) ignoring the saturation blocks, and 2) linearizing the AC power flow equations about an operating point set by (P2). The latter is done by finding the sensitivity of the head node active power of the feeder to the active power injection at the particular nodes where the VBs are situated. \textcolor{black}{The time-delay blocks are also ignored, but to account for time-delays, the controller design is done in such a way that the gain-crossover frequency of the open-loop system, or the bandwidth of the closed-loop system (without delays), is less than $1/5T_{d,\textrm{max}}$, where $T_{d,\textrm{max}}$ is the maximum time delay in the system. Effectively, this ensures that the delay-free system and the delayed system behave similarly.} Then, assuming that the system is excited by zero-mean wide-sense stationary Gaussian inputs, we minimize the sum of the variance of the error, $P_e$, denoted by $\sigma_{P_e}^2$, and a weighted sum of the variances of the control inputs to each VB, $P_{u1}$, $\dots$, $P_{un}$, denoted by $\sigma_{P_{u1}}^2$, ..., $\sigma_{P_{un}}^2$ respectively: 
\begin{equation*}\label{eq:nomrp}
\begin{array}{cc}
\min_{K} & \sigma_{ P_{e}}^{2}(K)+\sum_{r=1}^{n}\rho_r\sigma_{P_{ur}}^{2}(K),
\end{array}
\end{equation*}
where $K = [\begin{array}{cccc}
K_{1} & K_{2} & \ldots & K_{n}\end{array}]^{\top}$.
Here, $\sigma_{P_e}(K)$ and $\sigma_{P_{ur}}(K)$ are related to the standard deviation of the reference via the $\mathcal{H}_2$ norm of the transfer function, and to the gains $K_r$ through the Lyapunov equation, $A\Sigma+\Sigma A^{\top}+BB^{\top}=0$, where $A,B$ are system state matrices and $\Sigma$ is the state covariance matrix. See \cite{Cevat} for details. $\rho_r>0$ is a constant penalty parameter that we design to be inversely proportional to the power capacity of the $r$th VB. This penalizes power extraction from VBs that have a lower capacity to provide power output, thus resulting in a \textit{constraint-aware} controller. \textcolor{black}{The above nonlinear optimization problem can be solved efficiently (typically within six seconds for 30 VBs or less) using a path following algorithm \cite{hassibi1999path}.}

Note that there is a trade-off to be considered while choosing these constants. If $\rho\ll 1$, the control effort is less penalized, and hence, the resulting gains would be large (allowing larger control effort), and thus, VBs would saturate before they can assist in recovering the total head node power during a major disturbance. On the other hand, if $\rho\gg 1$, the control effort is penalized, and hence the resulting small gains may not be sufficient for better disturbance rejection capability. Thus, these constants should be chosen by taking into account this trade-off. 


 As a VB's energy state approaches its full capacity, the charging rate should be proportionately reduced, and when the energy state becomes empty, the discharging rate should be proportionately reduced. This helps to avoid a sudden step-change in power to zero when the VB saturates (either empty or full capacity).
To achieve this, the proportional VB controller gains, $K_r$, are then passed through an adjustment factor $K_{\textrm{adj,}r}$ (similar to standard gain scheduling) that is updated about every minute as follows:

{\small
\begin{align}\label{eq:adj}
K_{\textrm{adj,}r}=\begin{cases}
\begin{cases}
\left(\frac{B_{r}-\underline{B_{r}}}{\underline{B_{rl}} - \underline{B_{r}}}\right)^w, & \text{if }\underline{B_{r}}\leq B_{r}<\underline{B_{rl}}\\
1, & \text{else if }\underline{B_{rl}}\leq B_{r}\leq\overline{B_{r}}
\end{cases}, & \text{and } p_{b,r}<0\\
\begin{cases}
\left(\frac{\overline{B_{r}}-B_{r}}{\overline{B_{r}}-\overline{B_{rl}}}\right)^w, & \text{if }\overline{B_{rl}}<B_{r}\leq\overline{B_{r}}\\
1, & \text{else if }\underline{B_{r}}\leq B_{r}\leq\overline{B_{rl}}
\end{cases}, & \text{and } p_{b,r}>0\\
1, & \text{else if } p_{b,r}=0
\end{cases}
\end{align}
}where $\underline{B_r}$ and $\overline{B_r}$ represent the lower and higher bounds on the state of charge of the $r$th VB respectively, $B_{r}$ represents the state of charge of the $r$th VB, $p_{b,r}$ is the power output of the $r$th VB, $w$ is a constant greater than 1, and $\underline{B_{rl}}$ and $\overline{B_{rl}}$ are some chosen constants such that $\underline{B_{r}}<\underline{B_{rl}}<(\underline{B_{r}}+\overline{B_{r}})/2<\overline{B_{rl}}<\overline{B_{r}}$.
\vspace{-10pt}

\subsection{Overview of Inter-Feeder Control}
\begin{figure}[t]
\includegraphics[width=1\columnwidth]{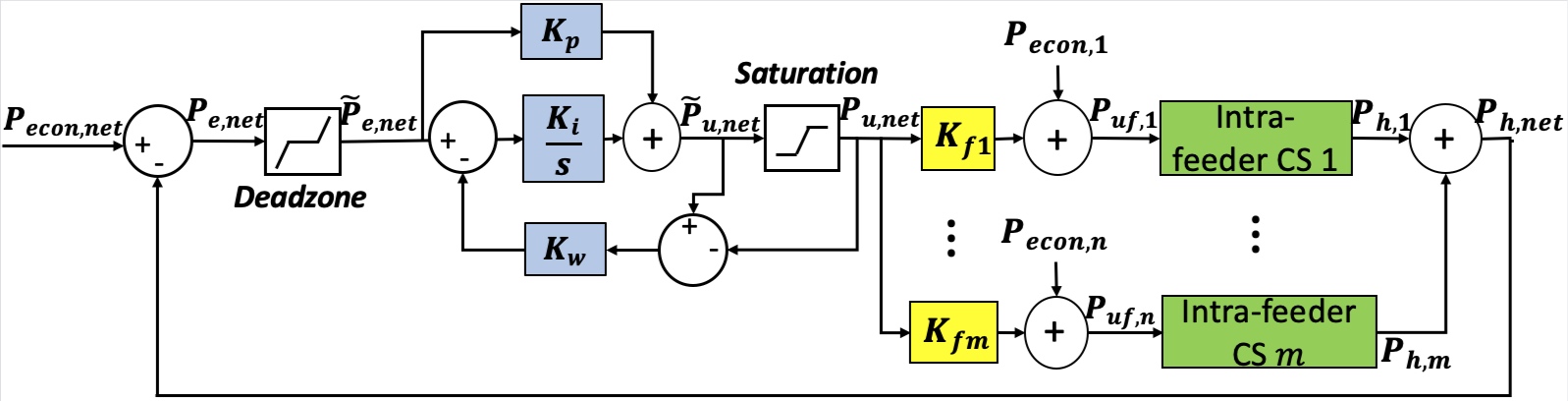}
\vspace{-20pt}
\caption{Inter-feeder Control}
\label{fig:interfeeder}
\vspace{-10pt}
\end{figure}

The intra-feeder control action described above is useful for mitigating small nodal disturbances that arise inside one feeder. However, to mitigate larger disturbances, such as if many of the VBs in one feeder saturate, the VBs inside one feeder may not be sufficient, and the flexibility of VBs in all the feeders will be required to reject the larger disturbance until (P2) updates the set-points.
The block diagram of the real-time control system for mitigating these inter-feeder disturbances is shown in Fig. \ref{fig:interfeeder}. \textcolor{black}{It is essentially a PI control scheme that corrects the economic set-point references to the $m$ intra-feeder control systems, the structure of each of which is shown in Fig. \ref{fig:intrafeeder}.} To ensure that the PI control is effected only when there is an appreciable change in the total head node power, a dead zone is also implemented before the PI controller. Moreover, a standard anti-windup mechanism is implemented inside this control scheme. This ensures that if all VBs inside all the feeders saturate in power, the PI control action is cut-off to avoid any integral windup. \textcolor{black}{This control scheme requires measurements from the head node powers from all the feeders. Since there is communication overhead involved in receiving the head node powers from all the feeders, the inter-feeder control action is proposed to be effected only about every 5 s.}

\textcolor{black}{
The working of the inter-feeder control system is described as follows. The sum of measured head node active powers from all feeders, denoted by $P_{\textrm{h,net}}$, is compared with the total economic market set-point for all feeders, $P_{\textrm{econ,net}}$. Then, for each time step $k$ (executed about every 5 s), the control input to the $i$th intra-feeder control system, $P_{uf,i}[k]$, is computed as follows:
\begin{align*}
P_{uf,i}[k]=K_{fi}P_{u,\textrm{net}}[k]+P_{\textrm{econ},i}[k],
\end{align*}
where $P_{u,\text{net}}$ is the output from the saturation block:
\begin{align*}
P_{u,\textrm{net}}[k]=\textrm{sat}\left(\widetilde{P}_{u,\textrm{net}}[k]\right)=\begin{cases}
\overline{P_{u,\textrm{net}}}, & \widetilde{P}_{u,\textrm{net}}[k]>\overline{P_{u,\textrm{net}}}\\
\widetilde{P}_{u,\textrm{net}}[k], & \underline{P_{u,\textrm{net}}}\leq \widetilde{P}_{u,\textrm{net}}[k]\leq\overline{P_{u,\textrm{net}}}\\
\underline{P_{u,\textrm{net}}} & \widetilde{P}_{u,\textrm{net}}[k]<\underline{P_{u,\textrm{net}}},
\end{cases}
\end{align*}
and $\widetilde{P}_{u,\text{net}}$ is given by the integrator as 
\begin{align*}
\widetilde{P}_{u,\textrm{net}}\left[k+1\right]&=\widetilde{P}_{u,\textrm{net}}\left[k\right]+K_{p}\left\{ \widetilde{P}_{e,\textrm{net}}\left[k+1\right]-\widetilde{P}_{e,\textrm{net}}\left[k\right]\right\}\\
& +K_{i}\left\{ \widetilde{P}_{e,\textrm{net}}\left[k\right]-K_{w}\left(\widetilde{P}_{u,\textrm{net}}\left[k\right]-P_{u,\textrm{net}}[k]\right)\right\},
\end{align*}
where
\begin{align*}
{P}_{e,\textrm{net}}[k]&= P_{h,\textrm{net}}[k]- P_{\textrm{econ,net}}[k],\\
\widetilde{P}_{e,\textrm{net}}[k] &=
\begin{cases}
{P}_{e,\textrm{net}}[k], & \left|{P}_{e,\textrm{net}}[k]\right|>P_{ed}\\
0, & \left|{P}_{e,\textrm{net}}[k]\right|\leq P_{ed}
\end{cases}
\end{align*}
and ${P}_{e,\textrm{net}}$ is the tracking error, $P_{ed}$ is the deadzone limit, $\widetilde{P}_{e,\textrm{net}}$ is the output of the deadzone, $K_w$ is the anti-windup gain, $K_p$ is the proportional controller gain and $K_i$ the integral gain in the PI controller, $\widetilde{P}_{u,\textrm{net}}$ is the unsaturated output of the PI controller,   $\overline{P_{u,\textrm{net}}}$ is the upper limit of the sum of control inputs to the intra-feeder control systems, $\underline{P_{u,\textrm{net}}}$ is the lower limit, $K_{fi}$ is a scaling factor, and $P_{\textrm{econ},i}$ is the economic reference for the $i$th feeder. $\overline{P_{u,\textrm{net}}}$ and $\underline{P_{u,\textrm{net}}}$ are computed by assuming that the VBs are at their power limits.
}

To penalize the extraction of power from feeders with lower capacity to provide power, the inter-feeder gains $K_{fi}$ are chosen in proportion to the total capacity of the feeder with which they are associated. If the feeder $i$ has a higher capacity, $K_{fi}$ is assigned a higher value, and if it has a lower capacity, $K_{fi}$ is assigned a lower value. Specifically, $K_{fi}=\overline{P}_{i}/\overline{P}$, where $\overline{P}_{i}$ is the power capacity of the $i$th feeder, and $\overline{P}$ is the total power capacity of all feeders. 

The PI gains are chosen to satisfy the desired requirement of phase margin and settling time. A high phase margin ensures that the system remains robust to VB failure, while a low settling time of less than a minute ensures that any large disturbance is quickly rejected before the optimal dispatcher provides new set-points to VBs every minute. To do this, first, the system is linearized in the same manner mentioned in Section \ref{ss:ifc}. Then, the phase margin of the open-loop transfer function from $P_{e,\textrm{net}}$ to $\Delta P_{h,\textrm{net}}$, as well as the settling time of the entire closed-loop system due to the response to a step input applied at $\Delta P_{\textrm{econ,net}}$, are computed for different values of $K_p$ and $K_i$, and the PI gains for which the phase margin is high and settling time is low is selected.

Fig. \ref{fig:pmst} shows an example of this process of selection of PI gains on a two-feeder system with 2 VBs in each feeder. The values of $K_p$ and $K_i$ were varied, keeping their ratio fixed, to study their effects on the phase margin and settling time. It can be seen that when $K_p=K_i$ (blue curve), there are some values of $K_p$ for which the system has a low settling time ($<30$ s, as indicated by the green dotted vertical line) but the phase margin varies significantly from 20$^{\circ}$ to 80$^{\circ}$. In that region, there is a value of $K_p$ and $K_i$ for which the settling time is the lowest, and phase margin is above 60$^{\circ}$, indicating the optimal value for the PI gains. Other values of $K_p$ and $K_i$, as those obtained from the $K_i = 0.1K_p$ and $K_i=10K_p$ curves, lead to sub-optimal performance.



\begin{figure}
\centering
\vspace{-12pt}
\subfloat{\includegraphics[width=0.47\columnwidth]{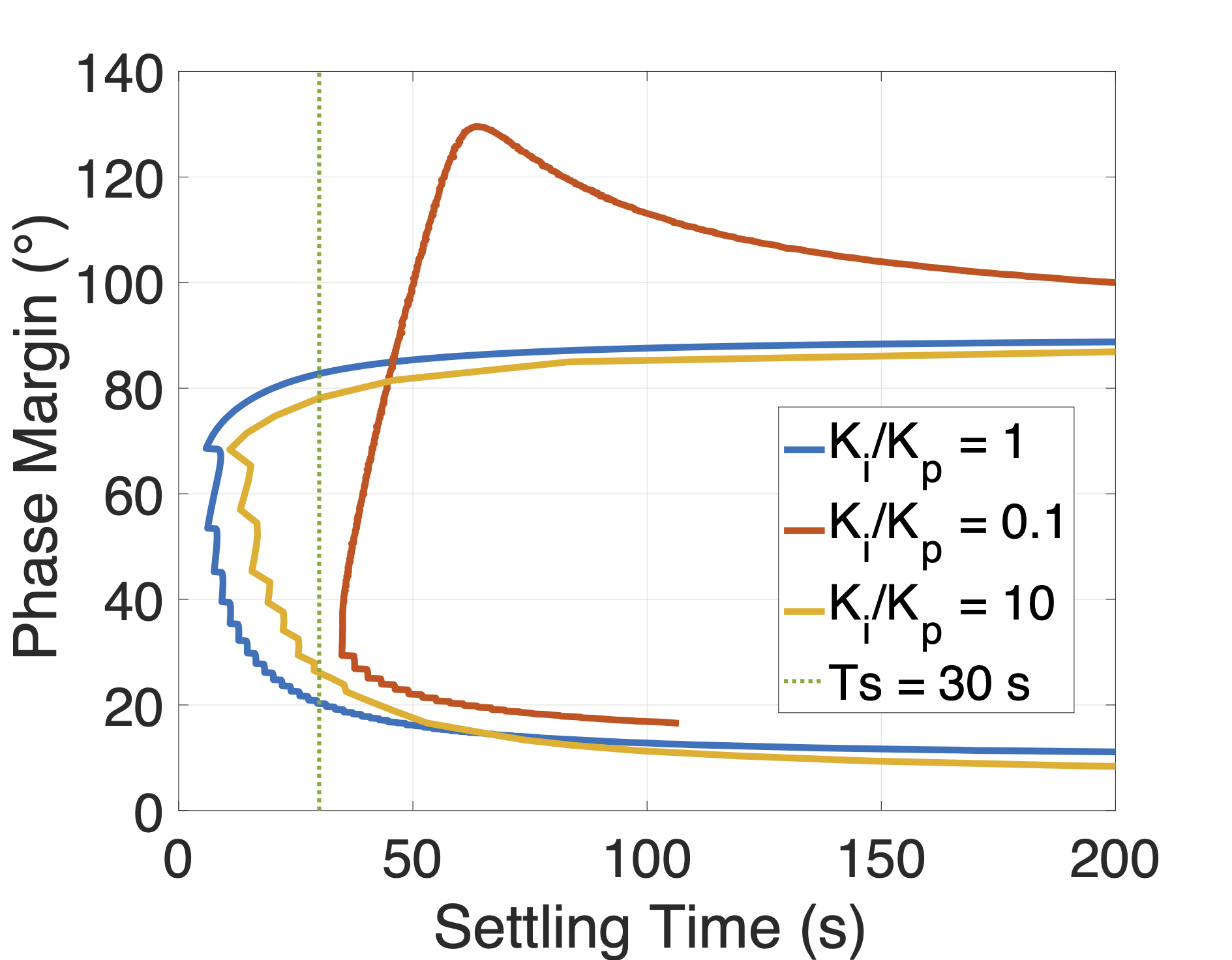}
\label{fig:pmst}}\hfill
\subfloat{\includegraphics[width=0.51\linewidth]{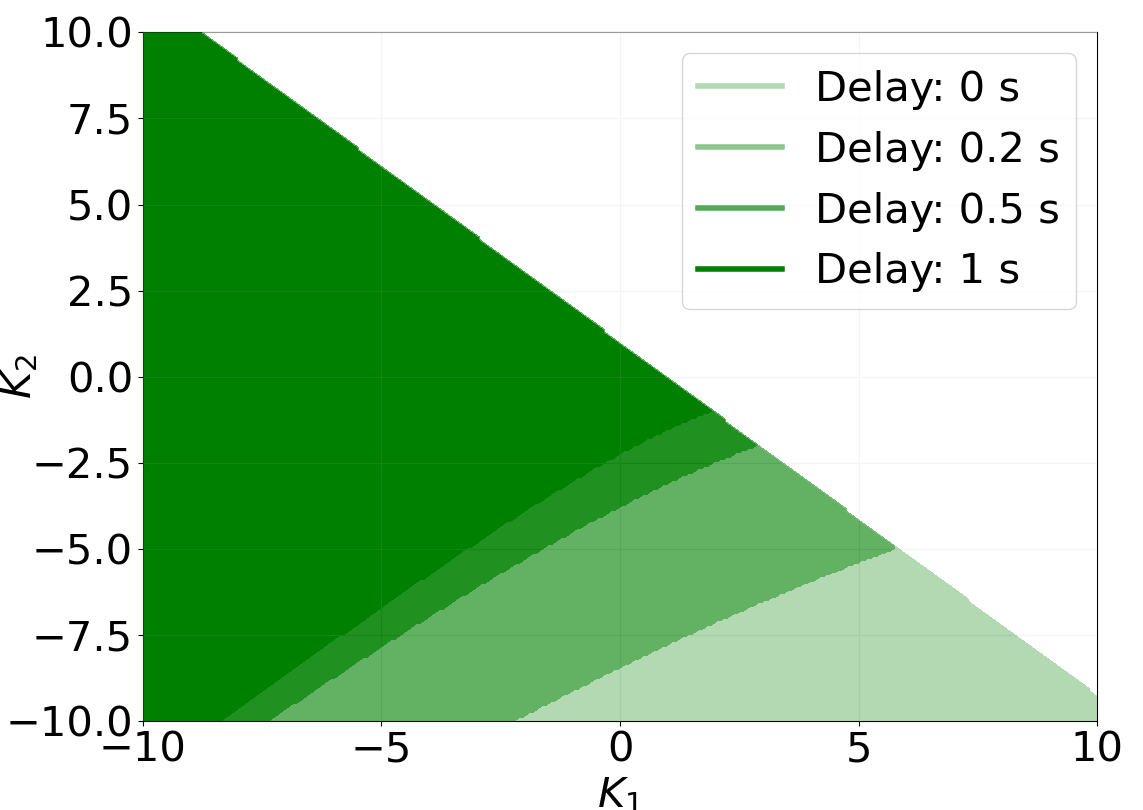}
\label{fig:delK}}
\caption{Region of Feasible Controllers. (\textit{Left}: a) Effect of PI gains on phase margin and settling time. The value of $K_p$ is varied from 0.01 to 24, and $K_i$ is then calculated according to the ratio $K_i/K_p$ mentioned in the legend. (\textit{Right}: b) Region of stable controller gains vs. delays.} 
\label{fig:pmstdelk}
\end{figure}
\subsection{Order and Frequency of Selection of Gains}\label{order}
 It is important to select/adjust the gains for both the intra-feeder and inter-feeder control mechanisms described above without hampering the real-time operation of the power system. If all the gains in both the inter-feeder and intra-feeder control mechanisms were chosen simultaneously, then the resulting optimization problem (if formed) would become intractable, hampering the real-time operation of the power system. Since it is important for VBs inside a feeder to respond to local disturbances quickly, without relying on the other feeders, we decouple the computation of the gains and design them in two stages: the proportional controllers inside each feeder are first chosen, considering each feeder to be separate, and then the PI controller to control all feeders is selected. This ensures the proper time-scale separation between intra-feeder and inter-feeder operation that was discussed before. Another benefit of this approach is that the optimization becomes more manageable as well.  Moreover, to ensure that the gains accurately reflect the current feeder structure, we propose that the gains be updated at the end of every 5 min. This renders our control scheme adaptive, of self-tuning regulator type. 
 {\scriptsize \begin{table}[!t]
    \centering
   
    \caption{Time required for PI Controller Tuning (Updated every 5 minutes)}
     \vspace{-10pt}
    \label{pitimes}
\begin{tabular}{ll}
\toprule
No. of feeders & PI Tuning Time (s) \\ 
\midrule
10             & 11.1         \\
5              & 5.6          \\
2              & 1.3          \\
\bottomrule
\end{tabular}
\vspace{-15pt}
\end{table}}
\textcolor{black}{Table \ref{pitimes} shows illustrative times (run on a laptop with 2.2 GHz Quad-Core Intel Core i7 processor and 16 GB RAM) required for the PI controller tuning every 5 min, as a function of the number of feeders. It can be seen that the time scales linearly with the number of feeders but is well less than 60 s. Real-time validation of this scheme on cyber-enabled devices through an OPAL-RT grid simulator is an ongoing work and will be the topic of a future publication.}


{\color{black}\subsection{Effect of Intra-Feeder Controller Gains and Delays on Stability}\label{sec:analysis}
In this subsection, an analysis of the effect of time-delays and controller gains on the stability of the control system is presented.
Consider the intra-feeder control diagram shown in Fig. \ref{fig:intrafeeder}, but with two controllable VBs (i.e., $n=2$). For analysis, let this system be linearized as mentioned in Section \ref{sec:rr}. The time-constants of the VBs are chosen to be $\tau_1=\tau_2=1$ s. \textcolor{black}{The region of stability was plotted (Fig. \ref{fig:delK}) using the system transfer function from $P_{\textrm{econ}}$ to $P_h$, replacing the delay with a third-order Pad\'e approximation (to extract poles).} Fig. \ref{fig:delK} shows the effect of a delay in the response of one of the VBs on the region of stability. \textcolor{black}{The delay in the second VB, $T_{d2}$, was varied from 0 to 200 ms, 500 ms, and 1 s (these values are chosen for the sake of illustration - in reality, delays may be smaller/larger), keeping no delay ($T_{d1}=0$) in the first VB. The regions of stability are shown (in Fig. \ref{fig:delK}) by different shades of green, with the lightest corresponding to no delay and the darkest corresponding to a delay of 1 s. Note that the regions depicted by the lighter shades of green \textit{include} those depicted by the darker.}

It can be seen that when there are no time-delays in either VB, the region of stability is a half-space. Indeed, it can be shown by analyzing the poles of the transfer function that the region of stability is given by $A_{p1}K_1+A_{p2}K_2>-1$ (this provides an analytic bound of permissible controller gains). It also indicates that there is a competition between the VBs: if one of the controller gains increases (and is hence more responsive to changes in head node power), the other must decrease to ensure stability. With time-delays, however, the region of stability is no longer a half-space but reduces to a triangular region (in the case of 2 VBs). Furthermore, this reduction is larger for VBs with greater communication delays. 

Thus, the above analysis indicates that for stability, allowable controller gains are limited by 1) the presence of other VBs, which introduce competitive behavior, 2)  communication delays, and that only a more restrictive set of gains is permissible for VBs with more communication delays. For an arbitrary grid, once the structure is known, one can do a similar analysis and determine the range of allowable controller gains that maintain system stability, which would be important for the DSO.}

\section{Simulation Results}\label{sec:sim}

\addtolength{\tabcolsep}{-0.5pt}
{\footnotesize\begin{table*}[t]
\caption{System Structure}
\vspace{-10pt}
\label{sysstruc}
\begin{center}
\begin{tabular}{llllll}
\toprule
Feeder & VB Locations & Active Power Noise Locations & Reactive Power Noise Locations & VB Time Constants, $\tau$ (s) & VB Delays, $T_{d}$ (s)\\
\midrule
1 & 712,722,706,703,727,708 & 705,712,727,728,730 & 712,718,707,727,729 &  1.0 1.8 1.4 1.6 1.9 1.5 & 0.8,0.8,0.5,0.7,0.7,0.3\tabularnewline
2 & 702,722,724,729,708,732 & 713,704,732,737,710 & 701,724,706,703,744,741 & 1.8,1.7,1.9,1.8,1.6,1.8 & 0.7,0.6,1.0,0.9,1.0,0.2\tabularnewline
3 & 701,705,703,730,734,741 & 701,705,704,714,727 & 705,707,725,727,740,736 & 1.3,1.1,1.3,2.0,1.8,1.6 & 0.9,0.5,0.3,0.9,0.4,0.6\tabularnewline
4 & 712,713,718,732,738 & 713,704,707,703,732,711 & 702,713,706,708,738,735 & 1.0,1.4,1.2,1.1,1.5 & 1.0,0.8,0.8,0.8,0.1\tabularnewline
5 & 724,709,731,737,736 & 712,709,737,738,740,735 & 714,725,744,731 & 1.8,1.7,1.1,1.3,1.9 & 0.3,0.5,0.3,0.1,0.5\tabularnewline
6 & 712,742,713,703,710 & 702,742,707,711,710 & 718,720,706,727,737,710 & 1.4,1.7,1.0,2.0,1.4 & 1.0,0.5,0.5,0.5,0.9\tabularnewline
7 & 722,744,728,738,711,736 & 701,705,714,707,732,711 & 705,744,731,733,711,710 & 1.1,1.6,1.7,1.30,1.5,2.0 & 0.5,0.1,0.6,0.6,0.8,0.7\tabularnewline
8 & 724,737,738,710,735 & 713,704,714,744,736 & 714,720,730,709,711 & 1.7,1.5,2.0,1.5,1.3 & 0.5,0.9,0.6,0.8,0.6\tabularnewline
9 & 701,718,737,741,735 & 724,703,744,731,708,710 & 742,724,709,737,710,735 & 1.7,1.5,2.0,1.1,1.1 & 0.8,0.5,0.7,0.4,0.9\tabularnewline
10 & 712,713,707,711,736 & 720,727,728,709,733,741 & 731,737,711,710,736 & 1.3,1.4,1.3,2.0,1.6 & 0.8,0.5,0.7,0.4,0.9\tabularnewline
\bottomrule
\end{tabular}
\end{center}
\vspace{-20pt}
\end{table*}}

To test the real-time VB corrective control mechanism along with the optimal tracking, 10 IEEE-37 node feeders (single-phase equivalents) were simulated (Fig. \ref{fig:ps}), with VBs at certain randomly picked locations inside each feeder, and active and reactive power noise sources at certain randomly picked locations, as shown in Table~\ref{sysstruc}.
\textcolor{black}{Both the OPF presented in Section \ref{sec:convex_form} and the real-time control mentioned in Section \ref{sec:rr} are simulated together in this section, with the real-time control using the full VB model \eqref{eq:batt_update}-\eqref{eq:batt_plimit} at every time-step and the OPF using the reduced VB model \eqref{eq:B_1}-\eqref{eq:B_3} every minute.} The parameters $\tau$ and $T_d$ of the VBs are also listed in Table~\ref{sysstruc}. Each feeder was assumed to have the same nominal active and reactive power demand that is present in the first phase of the standard IEEE-37 node feeder (which is about 727 kW in active power demand and 357 kVAR in reactive power demand). A unity power factor is assumed for the VBs. \textcolor{black}{The sample time for the simulations was assumed to be 500 ms.} The anti-windup gain was assumed to be 1, and the dead zone limit in the tracking error to be 72.7kW, which is 10\% of the total base demand in each feeder. The active power limits of the VB was set to $\pm 24.2$ kW when there were 6~VBs in the feeder, and at $\pm 29$ kW when there were 5~VBs in the feeder, such that the total capacity of the VBs in each feeder represents 20\% of the nominal demand. The energy bound of each VB was assumed to be $97$ kWh (where there were 6 VBs in the feeder) and 116 kWh (where there were 5 VBs), thus, representing a maximum 4-hour continuous operation of each VB at its power limit.  This is a reasonable scenario for many DERs, including water heaters and residential batteries. The controllers were designed according to the process mentioned in Section~\ref{sec:rr}.

	\begin{figure}[t]
		    \centering
		    \includegraphics[width=\columnwidth]{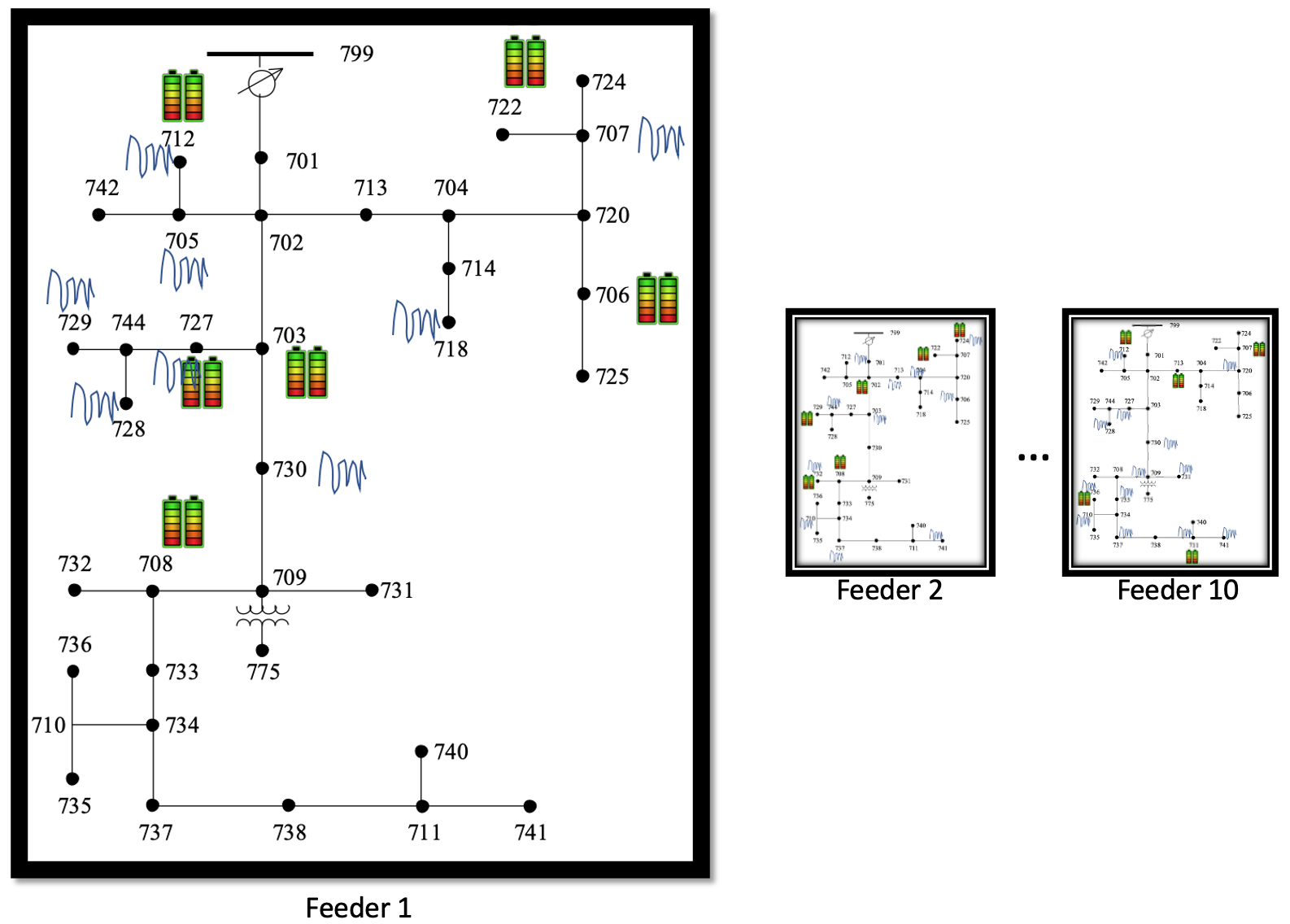}
		     \vspace{-10pt}
		    \caption{Power System Structure}
		    
		    \label{fig:ps}
		    \vspace{-10pt}
		\end{figure}

\begin{figure}[t] 
\centering 
    \includegraphics[width=0.8\columnwidth]{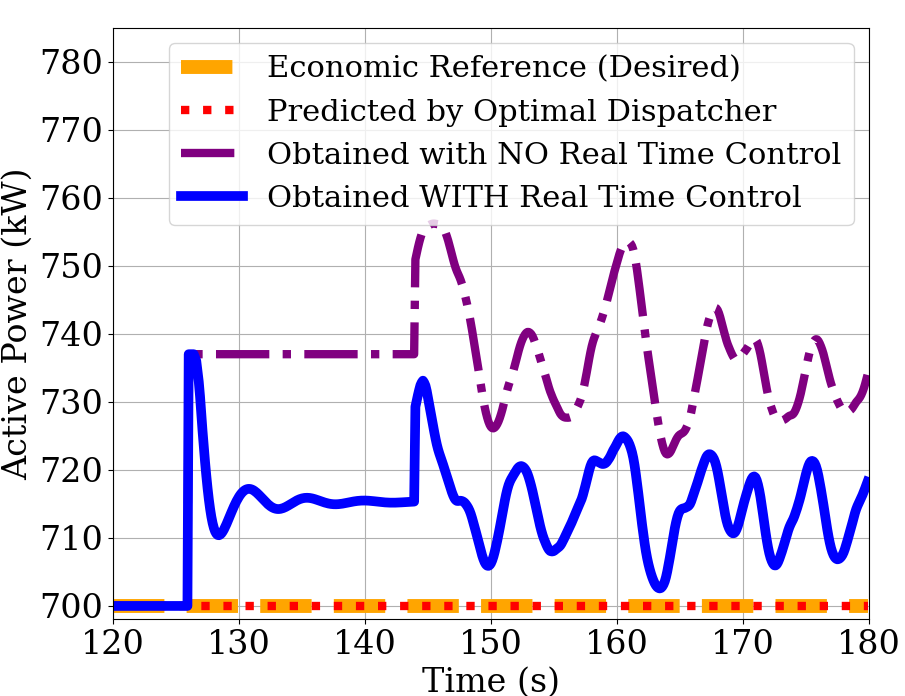}
\caption{Intra-feeder disturbance rejection}
\label{fig:intrasim}
\end{figure}
\subsection{Intra-feeder disturbance rejection scenario}
We first illustrate the rejection of intra-feeder disturbances such as fluctuations in solar PV, which represent clouds obstructing solar PV panels. Since (P2) operates every minute, such disturbances would have to be mitigated within a minute. Also, since such a disturbance is local, other feeders should not be actuated to reject such a disturbance. The results for a 1-minute simulation are shown in Fig. \ref{fig:intrasim}. After small modeling errors were corrected for by two consecutive runs of (P2) at $t=60$ s and $t=120$ s, a step disturbance was introduced at around 125 s into all the nodes of Feeder 1 at the active and reactive power noise locations shown in Table~\ref{sysstruc}. This step disturbance can represent, for example, a big cloud cover obstructing the sun. From 145 s to 180 s, random noise was injected into the same nodes instead of a step disturbance. This can represent small clouds intermittently covering the sun. The blue curve denotes the total actual head node power from all feeders. The orange dashed line depicts the economic market set-point desired to be met by all the feeders. The total head node power predicted by employing VB set-points resulting from solving (P2), every minute, is shown by the red dotted line. The optimization horizon for (P2) is 10~minutes. The purple dash-dotted line shows the head node power if the VBs were optimally dispatched, but no real-time control mechanism, as described in Section \ref{sec:rr}, was in place. The results show that in the presence of a step disturbance, the proportional gains inside the feeder try to bring the power back towards the economic set-point. However, due to the nature of proportional control, there is a steady-state error. Moreover, if the noise is random, the intra-feeder control mechanism reduces the standard deviation of the random noise, as in Fig. \ref{fig:intrasim}, where the standard deviation reduces from 10.8 kW to 8.5 kW, a reduction by 26\%. Note that inter-feeder PI control is not activated in the above simulation since the disturbances are all within the dead zone of 72.7 kW.

\subsection{Inter-feeder disturbance rejection}
\begin{figure}[t] 
\centering 
\subfloat[Total Head node power]{\includegraphics[width=0.8\columnwidth]{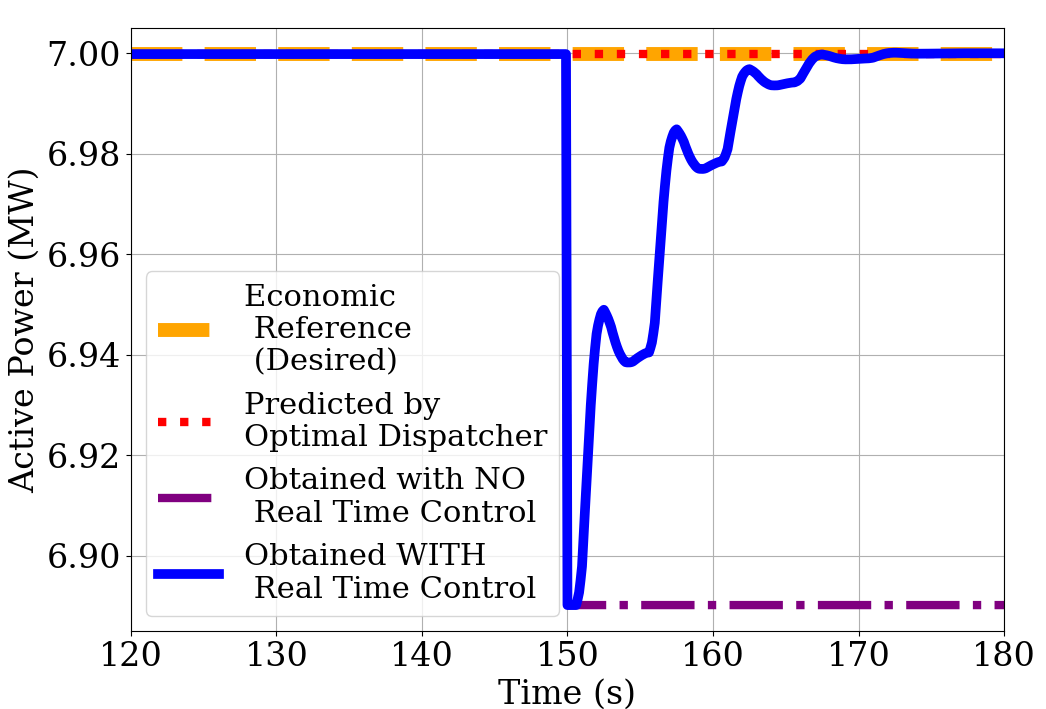}
 \label{fig:intersim}}
\hfil
\subfloat[Head node power from 9 feeders with no VB saturation]{\includegraphics[width=0.8\columnwidth]{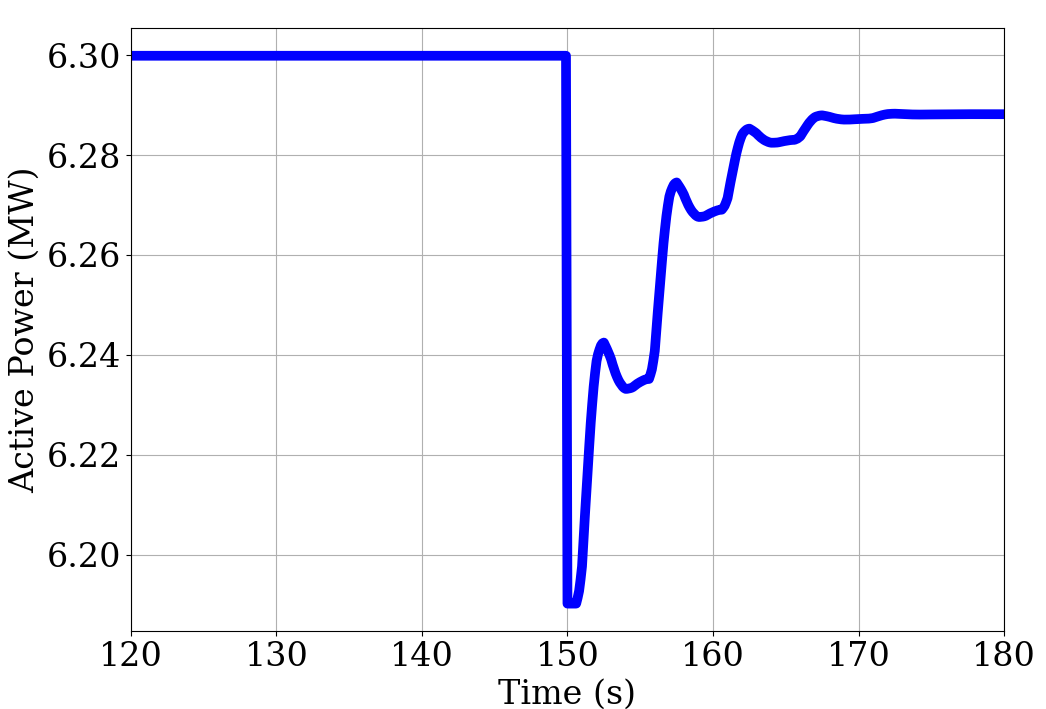}
\label{fig:intersim9other}}
\caption{Inter-feeder disturbance rejection scenario}
\vspace{-10pt}
\end{figure}
In this subsection, we illustrate the effect of rejecting an inter-feeder disturbance and thus, the resilience of the VB coordination framework. It is assumed that due to an adversarial cyberattack on a feeder, all the VBs in that feeder are set to their power limits.
Fig.~\ref{fig:intersim} shows the results of a 5-minute simulation with optimal PI controllers. At $t=150$ s, one of the ten feeders is attacked in an above-mentioned manner.  As a result of this attack, the power delivered by VBs suddenly change from 36.1 kW to 145.4 kW (an increase by 302.7\%). It is seen that PI control brings the head node power back to the economic set-point. Because other feeders also participate in this major disturbance, the head node power from them also changes after the disturbance, as illustrated in Fig.~\ref{fig:intersim9other}.

\begin{remark}[Voltage values]
In both the intra-feeder and inter-feeder disturbance rejection, the control mechanism successfully keeps the voltages within acceptable limits. In the examples shown above, voltages in the nodes of the feeders are all above 0.97 pu, with the maximum voltage being 1.0 pu (which is the fixed head node voltage that assumes reactive power support from tap changing transformers is separately available).
\end{remark}
\vspace{-20pt}
\textcolor{black}{
\subsection{Integration with Multiple runs of OPF}
In this subsection, a 2-min simulation is performed (Fig. \ref{fig:mfr}) to show the full hierarchy of the OPF and the real-time corrective control action in the case of VBs undergoing cyberattack. Specifically, two of the VBs in both the 9th and 10th feeder are assumed to be saturated (via attack) after 15 s, resulting in the PI control mechanism being activated, and obtained power (blue line) being restored to the desired value (yellow line). Note that without real-time control (purple dash-dotted line), restoration would not have been achieved. At the next run of the OPF at 1 min, however, even with no real-time control, the optimal set-point dispatcher adequately re-dispatches remaining VBs to result in power close to the desired power. Note that the spike in the blue line at 1 min exists because the VB set-points change, and so does the net real-time control action (which is the sum of set-point provided by the optimal set-point dispatcher and the corrective control).}
\begin{figure}
    \centering
    \includegraphics[width=0.9\columnwidth]{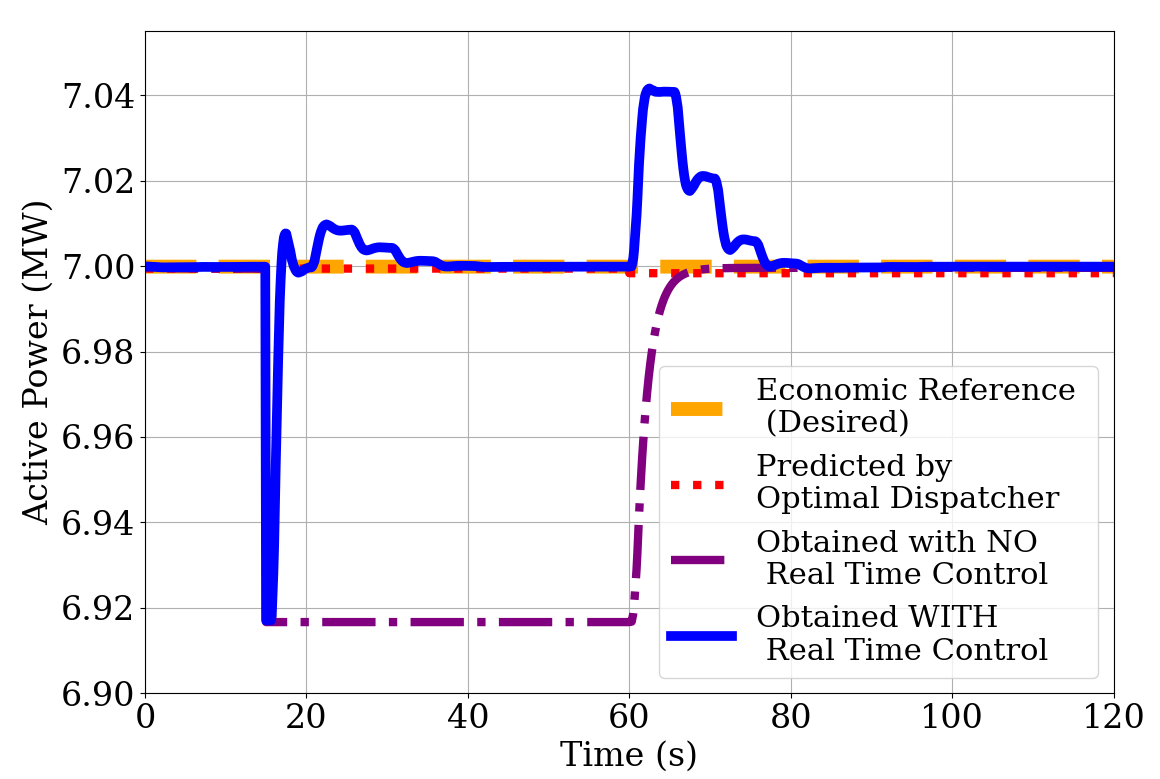}
    \caption{Simulation of the real-time controller integrated with multiple runs of OPF}
    \label{fig:mfr}
\end{figure}

\vspace{-5pt}
\section{Conclusions and Future work}\label{conc}
In this paper, analysis and simulation results have been presented in support of a novel framework for large-scale coordination of DERs to support deep penetration of renewable energy. The explicit consideration of (temporal) energy and (spatial) grid constraints and the economic and reference-tracking (techno-economic) objectives have been achieved via a spatio-temporal decomposition approach that leverages information on demand-side flexibility to disaggregate grid economic trajectory into reference control signals for virtual batteries in distribution feeders. A convex optimal power flow (OPF) formulation has been presented that ensures a provably tight optimal dispatch of virtual batteries (VBs) to track an economic power trajectory. Resilient and real-time control techniques to adapt and recover from intra-feeder and inter-feeder disturbances have been analyzed and designed. To show the effectiveness of the decomposition approach and the real-time control mechanism, simulation results have been conducted on a modified IEEE-37 test system. 

Future work will incorporate reactive power control of VBs and extend the regulation of voltage with inverters. Further, the market economic problem will be considered explicitly~\cite{ChengdaACC} to study the coupling between the market layer economic problem and the feeder constraint aware dispatch. Future work will also try to extend this approach to the three-phase unbalanced system operation~\cite{nazir_optimal3phase}. \textcolor{black}{Finally, we are interested in extending the multi-period (P2) to a robust formulation to trade off conservativeness of dispatch and the probability of voltage and VB violations~\cite{nazir2020stochastic}.}


\section{Acknowledgements} 
{\color{black} We are grateful for the four anonymous reviewers who helped improve the paper} and many helpful discussions with Enrique Mallada and Dennice Gayme from Johns Hopkins University (on economic reference signals); Soumya Kundu and Thiagaran Ramachandran from PNNL (on VB models and DSSE); Pavan Racherla from UVM (on grid modeling); {\color{black} and Kerrick Johnson from VELCO and Pete Maltbaek from Smarter Grid Solutions (on real-time implementation and grid communications).} 

\vspace{-5pt}
\bibliographystyle{IEEEtran}
\small\bibliography{fix.bib}
\end{document}